\documentclass[mnsc,nonblindrev]{informs3_nojournal} 

\OneAndAHalfSpacedXI
\usepackage{booktabs}
\usepackage[backref=page]{hyperref}
\usepackage[capitalise]{cleveref}
\usepackage{algorithm,algorithmic}
\usepackage{tikz}
\usepackage{enumitem}
\usetikzlibrary{automata, positioning}
\usepackage{subcaption}
\newcommand{\ep}{\epsilon}

\usepackage{enotez}
\let\footnote=\endnote

\usepackage{bbm}

\hypersetup{
    colorlinks=true,
    linkcolor=blue,
    filecolor=magenta,      
    urlcolor=cyan,
    pdftitle={Unit Demand Pricing for Min Buying},
    pdfpagemode=FullScreen,
    }
    
\def\bold#1{{\sffamily \bfseries#1}}

\def\E{\mathbb{E}}
\def\e{\mathbf{e}}
\def\I{\mathbb{I}}
\def\cI{\mathcal{I}}
\def\rev{\mathcal{R}}

\def\cF{\overline{\text{F}}}
\def\F{\text{F}}

\def\p{{\bf p}}
\def\bp{{\bf p}}

\def\bV{{\bf V}}
\def\v{{\overline{v}}}

\def\rev{\mathcal{R}}

\def\Tri{\text{Tri}}

\graphicspath{ {images/} }

\usepackage{natbib}
 \bibpunct[, ]{(}{)}{,}{a}{}{,}%
 \def\bibfont{\small}%

\TheoremsNumberedThrough %Preferred (Theorem 1, Lemma 1, Theorem 2)
\ECRepeatTheorems

\EquationsNumberedThrough    % Default: (1), (2), ...
\MANUSCRIPTNO{}

\begin{document}
%%%%%%%%%%%%%%%%

\RUNAUTHOR{\bold{Cui et al.}}
% \RUNAUTHOR{\bold{Blinded}}

\RUNTITLE{Unit-Demand Pricing for Min-Buying Customers}

\TITLE{\Large \bold{Unit-Demand Pricing for Min-Buying Customers}}

\ARTICLEAUTHORS{%
\AUTHOR{Titing Cui$^{1}$, Michael L. Hamilton$^{2}$, Jourdain Lamperski$^{3}$, Kasra Ganjizadeh Tari$^{3}$}
\AFF{$^{1}$University of Tulsa, $^{2}$Baruch College, $^{3}$University of Pittsburgh,  \EMAIL{\\tic3653@utulsa.edu, michael.hamilton@baruch.cuny.edu, lamperski@pitt.edu, kag339@pitt.edu}} 
}

\ABSTRACT{
% OR requires one paragraph, <200 word abstract. 
We study the unit-demand pricing problem under \textit{min-buying} behavior, where customers purchase the lowest-priced item that yields non-negative utility. In contrast to the quasi-linear model, we show that uniform pricing becomes substantially more effective under min-buying. First, we study symmetric environments. We show that for min-buying customers with independent and identically distributed valuations, a uniform price across items is optimal. We then extend this optimality result to all valuations drawn from exchangeable and affiliated distributions. However, uniform pricing is not optimal in all symmetric environments: for min-buying customers we construct an exchangeable and non-affiliated valuation distribution for which the gap between uniform pricing and optimal item pricing is logarithmic in the number of items. This gap is tight and matches the worst-case gap in the quasi-linear setting. Next, we consider asymmetric environments in which valuations are independent but not identically distributed. When these valuations are regular and customers are min-buying, we show that a uniform price guarantees approximately 69.86\% of the optimal item-pricing revenue and this is tight. Overall, our results highlight that the value of non-uniform pricing substantially diminishes, and often disappears entirely, under more conservative demand models in which customers engage in satisficing behaviors such as min-buying.
}

\KEYWORDS{unit-demand pricing, min-buying, algorithmic pricing, price discrimination} 

\HISTORY{\underline{First Version}: March 14, 2025. \underline{This Version}: \today}

\maketitle

% {\color{red} To do: 
% \begin{itemize}
%     \item Journal: Operations Research - Standard Paper
%     \begin{itemize}
%         \item Area: Markets, Platforms, and Revenue Management. 
%         \item DE: René Caldentey 
%         \item AE: Negin Golrezaei, Jake Feldman, Stefanus Jasin, Ming Hu (alt), Brendan Lucier (alt).
%         \item Reviewers: Billy Jin (bzj3@cornell.edu), Shixin Wang (GTech), Tarek Abdullah (Northwestern), Omar El Housni, Gerardo Berbeglia, Jake Feldman, Shreyas Sekar, Zhe Liu, Xiaobo Li
%     \end{itemize}
% \end{itemize}
% }

% OR-Short Paper Submission: A focused technical paper is a short paper that has results on a technical problem of broad interest to the OR community. The OPRE journal in the past has had success publishing shorter papers, and this category serves to specifically remind us that page length does not equate with importance. The maximum length for a focused technical paper is 20 pages (excluding references), and the Electronic Companion (EC) is for code and experiment details only (no proofs – proofs should be in the published portion of the paper).

% {\color{red} To do: 
% \begin{itemize}
% % \item[1. ] Call CER distributions, triangular-revenue distribution everywhere. Remove the canonical. 
% \item[2. ] Check and clean the proof of \cref{thm:indep_reg}.
% \item[3. ] Clean up writing in \cref{sec:reg}.
% \item[4. ] Submit to Operations Research on August 14th 
% \end{itemize}
% }

\section{Introduction}\label{sec:intro}
In a unit-demand pricing problem (UDP), the task is to find prices for $n$ items that maximize a seller's expected revenue given some distribution over customer valuations. Once prices are set, each unit-demand customer draws a vector of valuations, and purchases the item they prefer most (or nothing) according to some decision rule. The most commonly studied decision rule is when customers are \textit{quasi-linear utility-maximizers} who choose the item that maximizes the difference between their valuation and the price, or no item when all differences are negative. The next most common decision rule is called \textit{min-buying}, as popularized by \cite{rusmevichientong2006nonparametric}. Under min-buying, customers purchase the lowest priced item that gives them non-negative utility. 

Both quasi-linear utility-maximization and min-buying are useful modeling choices in different settings. While quasi-linear utility-maximization is an intuitive and natural way to model rational customer behavior, it depends on assumptions about how customers make choices that do not hold in all settings. Namely, it requires that customers translate their preferences into a common numerical scale and then choose based on utility comparisons between all pairs of items. Research in behavioral economics provides many examples where, due to limited attention, cognitive costs, and/or other biases, customers fail to maximize their quasi-linear utility in this way and instead are better modeled as following simple heuristics to make purchasing decisions (for a survey of results in this area, see \cite{kahneman2006anomalies}). Inspired by such settings, when customers are not modeled as fully rational, a common alternative is to suppose that customers engage in \textit{satisficing} behavior. Under satisficing behavior customers screen items by whether they are ``good enough'' (i.e., yield non-negative utility) and then pick the cheapest acceptable option \citep{kahneman2003maps}. When the screening is done in ascending price order this is equivalent to following a min-buying rule.

Beyond its use as a model of bounded rationality, min-buying is also closely connected to several important problems in revenue management and pricing. Min-buying rules are commonly used to model how customers make decisions under incomplete information. For instance, in settings with opaque products (i.e., products for which some attribute of the good, such as color or style, is hidden until after purchase) a common assumption is that customers are \textit{risk-averse}, meaning that they value the opaque product according to the minimum-valued possible allocation. Models of opaque selling for risk-averse customers are often related to min-buying UDPs \citep{elmachtoub2021power,chen2025approximation,housni2025price}. UDPs with min-buying customers are also structurally identical to so-called order-oblivious posted-price mechanisms (OPMs), in which a single good is offered sequentially to customers at buyer-specific prices and the mechanism must perform well regardless of the order in which the offers are considered \citep{chawla2010multi}, as well as to certain ranked assortment-optimization problems \citep{aouad2018approximability,berbeglia2020assortment}.

For either choice of decision rule, UDPs of the above form are foundational models for understanding the value of \textit{price discrimination}, which concerns setting different prices for horizontally differentiated items, and \textit{price differentiation}, which concerns setting prices for vertically differentiated items. Much recent work on UDPs has focused on quantifying the value of price discrimination or differentiation relative to the strategy of offering a single uniform price for all items. For quasi-linear customers, however, the story is mixed. On the one hand, in the natural setting of IID valuations and quasi-linear utility-maximizing customers, a uniform price need not be optimal \citep{thanassoulis2004haggling}, and optimal discriminatory pricing can earn up to twice as much revenue as uniform pricing \citep{chawla2010multi,dutting2016revenue}. (Note pricing in this setting is referred to as \textit{discriminatory} since, when valuations are IID, every permutation of the price vector yields identical revenue, thus the assignment of relatively high and low prices to particular items is arbitrary). On the other hand, a complementary line of work shows that uniform pricing can be remarkably effective even in asymmetric (i.e., vertically differentiated) settings. For quasi-linear utility-maximizing customers with independent but non-identical item valuations satisfying a regularity condition, an optimally chosen uniform price always earns at least $38\%$ of the optimal non-uniform pricing revenue \citep{jin2019tight}. This guarantee is surprisingly strong given how different the item distributions may be, potentially reflecting stark differences in item quality, and highlights the general effectiveness of uniform pricing.

Viewed through this lens, much of the literature on UDPs can be understood as quantifying when non-uniform prices are necessary. Because uniform prices are fair, easy to implement, and easy to optimize, they are attractive to sellers so long as restricting to a uniform price does not result in too much revenue loss. For quasi-linear utility maximizers, the performance of uniform pricing relative to optimal non-uniform pricing is well understood: prior work provides tight or nearly tight bounds in canonical regimes, including IID valuations, independent regular but non-identical valuations, and general worst-case distributions. However, the practical and managerial implications of these results are less clear, with relatively weak guarantees in natural symmetric settings and surprisingly strong results in asymmetric ones. Further, as these results are closely tied to the assumption that customers maximize quasi-linear utility, an assumption which may not hold in many practical settings, it is then natural to ask how these classical conclusions change when customers are min-buying rather than quasi-linear utility maximizers. In this paper, we answer this question by proving analogues of these foundational results under min-buying behavior. In contrast to the mixed results in the quasi-linear case, we find that uniform pricing is optimal in many symmetric environments and, in general, achieves substantially stronger guarantees.

\subsection{Our Contributions}

A summary of our main contributions is as follows:

\begin{itemize}
    \item[1. ] First, we study the value of discriminatory pricing in symmetric valuation environments when customers are min-buying. We consider three classes of symmetric valuations: IID distributions, affiliated and exchangeable distributions, and general exchangeable distributions. We show that when customers are min-buying and valuations are IID, a uniform price is always optimal (cf. \cref{thm:iid}) and thus discriminatory pricing has no additional power in this setting. We then extend this result to the more general class of affiliated and exchangeable distributions and show uniform pricing continues to be optimal (cf. \cref{thm:exchang_aff}). Finally, we complement these positive results with a lower-bounding construction for general exchangeable valuations that demonstrates, both in the min-buying and quasi-linear settings, that discriminatory pricing can earn logarithmically more revenue than uniform pricing (cf. \cref{thm:exchang}) matching the known upper-bound. 
        
    \item[2. ] Second, we study the value of differentiated pricing when valuations are drawn asymmetrically across items. Our main result in this setting is to show that for min-buying customers, when valuations are drawn from independent and (not necessarily identical) regular distributions, a uniform price is guaranteed to earn $\approx 69.86\%$ of the revenue of the optimal non-uniform pricing and this is tight (cf. \cref{thm:indep_reg}). This resulting approximation substantially improves over the $\approx 38 \%$ that can be guaranteed in the quasi-linear case, further emphasizing the differences between these two models. 
\end{itemize}

Overall, our work highlights how the unit-demand pricing problem changes when customers are min-buyers versus quasi-linear utility-maximizers. We situate our results within the broader literature on unit-demand pricing and show that, both for price discrimination among horizontally differentiated items and for price differentiation across (potentially) vertically quality-differentiated items, the gains from non-uniform pricing depend critically on the customer decision rule, and are markedly smaller under min-buying (often vanishing altogether). En route to our results, we also develop techniques for how to analyze pricing in general symmetric settings and how to exploit the structure of min-buying behavior which may be of independent interest.

\subsection{Literature Review}\label{sec:lit}
The literature on unit-demand pricing is vast and interdisciplinary, with contributions from economics, operations research/management, marketing, and computer science. In this subsection, we describe some of the most relevant streams of work on min-buying rules and various unit-demand pricing problems, and connect these results with our findings.

\subsubsection*{Models of Min-Buying}
Min-buying models of customer behavior have a long history, going back at least to \citet{simon1955behavioral}, who highlights the potential value of alternative models of customer decision-making, including satisficing rules based on simple min/max criteria. More recently, behavioral evidence that quasi-linear maximization can, in some settings, provide a poor explanation of customer choice has spurred renewed interest in min-buying models \citep{gigerenzer1996reasoning,kahneman2003maps}. These min-buying-style models are also supported by empirical research in several contexts. For example, \citet{stuttgen2012satisficing} use eye-tracking data to show that many customers engage in satisficing behavior (rather than comparing across all available items) and \citet{karle2015loss} provide theoretical and empirical evidence that customers may exhibit loss aversion and gravitate toward lower-priced items more strongly than would be expected under a fully rational, quasi-linear utility-maximizing model.

Building on their use as a model of bounded rationality, a growing body of work studies pricing and revenue-management problems under min-buying-type assumptions (e.g., \cite{goyal2016near,cao2022network,bhardwaj2025algorithms}). Min-buying rules are also commonly used to model customer decisions under incomplete information, such as arises in settings with \textit{opaque products}, where information about the good is hidden until after purchase. In models of opaque selling, it is common to assume that customers are \textit{risk-averse}, meaning that they value an opaque product according to the minimum-valued possible allocation \citep{elmachtoub2021power,chen2025approximation,housni2025price}. Furthermore, unit-demand pricing under min-buying is closely connected to certain assortment-optimization models \citep{aouad2018approximability}. In particular, \citet{berbeglia2020assortment} establish an exact correspondence between min-buying UDPs and revenue-ordered assortments, allowing our results to translate directly into new guarantees for specific instances of their model.

\subsubsection*{UDPs in Symmetric Settings}
There is a deep body of work studying the value of price discrimination in symmetric settings. When valuations are drawn IID and customers are quasi-linear utility maximizers, a stream of papers \citep{thanassoulis2004haggling,chawla2007algorithmic,chawla2010multi,dutting2016revenue} shows that non-uniform discriminatory prices can strictly outperform uniform pricing. For $n$ items, the ratio between the revenue of the optimal item pricing and that of the optimal uniform pricing can be as large as $2-\frac{1}{n}$, and this bound is tight. Beyond their direct implications for pricing, these results have also been influential in the algorithmic game theory literature because of their close connections to prophet inequalities, which have spurred substantial follow-up work; see, for example, \cite{lucier2017economic,correa2019pricing}.

Beyond IID valuations, prior work has studied other symmetric settings in which valuations are correlated across items and exhibit some form of positive dependence \citep{anunrojwong2022robustness}, as well as the general setting in which valuations are exchangeable. In particular, for general exchangeable valuations and quasi-linear utility-maximizing customers, \cite{elmachtoub2021power} show that discriminatory pricing can earn at most a logarithmic factor more revenue than uniform pricing and ask whether this bound is tight. We answer this question in the affirmative by constructing an exchangeable instance in which, under min-buying, discriminatory pricing earns a $\Theta(\log n)$ factor more revenue than uniform pricing, where $n$ is the number of items. This lower bound also applies to quasi-linear customers since, for any fixed price vector, revenue is weakly higher under quasi-linear utility maximization than under min-buying, whereas uniform-pricing revenue is identical under the two decision rules. Conversely, every upper bound for quasi-linear customers applies to min-buying UDPs. Together, these observations establish that the logarithmic gap is tight up to constant factors under both decision rules.

In parallel with work characterizing approximation guarantees for uniform pricing, a stream of papers studies the algorithmic problem of solving UDPs, sometimes referred to as \textit{item pricing}, in symmetric settings. For quasi-linear utility-maximizing customers with IID valuations, \cite{chen2014complexity,chen2018complexity} show that finding the optimal item pricing is NP-hard. On the positive side, when valuations are IID and bounded, \cite{cai2015extreme} provide an additive PTAS, while \cite{bowers2025polynomial} give a multiplicative PTAS. In this paper, we study the analogous problem for min-buying customers and show that uniform pricing is always optimal when valuations are IID. Our result simultaneously characterizes the performance of uniform pricing and reduces the corresponding item-pricing problem to a one-dimensional optimization.

\subsubsection*{UDPs Beyond Symmetric Settings}
Outside symmetric settings, there is also a vast body of work studying the value of price differentiation. The most relevant stream is that of \citet{alaei2019optimal,jin2019tight,jin2020tight}, which show that when customers are quasi-linear utility maximizers and valuations are drawn independently but non-identically from regular distributions, the revenue of an optimal uniform price is always within a factor of approximately $2.62$ of the optimal differentiated-pricing revenue, and that this factor is tight. In our work, we study the analogous problem for min-buying customers and characterize the exact worst-case approximation factor, showing that an optimal uniform price is within a factor of approximately $1.43$ of the optimal differentiated pricing. Equivalently, our result characterizes the tight revenue gap between anonymous pricing and optimal order-oblivious posted pricing for independent regular valuations (and alternatively, the result of \cite{jin2019tight} can be cast as giving the tight approximation factor between differentiated pricing and uniform pricing when customers are \textit{max-buying}, meaning they purchase the highest priced item for which they obtain non-negative utility). On the algorithmic side, there is a significant literature on solving or approximating the optimal item-pricing problem when valuations are non-identical and customers are quasi-linear utility maximizers; see \cite{aggarwal2004algorithms,guruswami2005profit,balcan2008item,gallego2021bounds}. Several of these techniques and guarantees also extend to related min-buying formulations. However, relative to this literature, our focus is on establishing tight constant-factor approximation guarantees for the simpler policy class of uniform pricing.

% Finally, we highlight one interesting setting of differentiated pricing where two valuations are drawn independently according to shifted uniform distributions. This setting is a natural testbed to showcase ideas in algorithmic game theory, see for instance \cite{pavlov2011optimal,daskalakis2015strong}. In this work we also highlight this setting an show in this case, uniform pricing is revenue optimal for min-buying customers. 

\subsection{Paper Outline}
The remainder of the paper is organized as follows. In \cref{sec:model}, we describe the model and give some preliminary results about unit-demand pricing for min-buying customers. In \cref{sec:symm}, we study the performance gap between uniform pricing and discriminatory pricing when valuations are drawn symmetrically. In \cref{sec:reg}, we consider when items are valued asymmetrically and bound the gap between uniform pricing and differentiated pricing. Finally, in \cref{sec:conclusions} we conclude and describe avenues for future work. All omitted examples and proofs referenced in the main body can be found in \cref{sec:examples,sec:proofs} in the Appendix.

\section{Model and Preliminaries}\label{sec:model}
We consider a seller offering $n$ items, described by the index set $[n] := \{1, 2, \ldots, n\}$, to customers with non-negative, random valuations $\bV = (V_1, \ldots, V_n)$. For each item $i$, we will denote the marginal valuation distribution by $V_i \sim \F_i$. Unless otherwise stated, we suppose each marginal distribution $\F_i$ is absolutely continuous with density $f_i$\footnote{Continuity is not strictly necessary for our results. We assume it for ease of exposition, as it allows us to avoid complications that occur when valuations exactly equal prices with positive probability.}. We denote the survival function (cCDF) of a distribution $\F_i$ by $\cF_i(x) := 1 - \F_i(x)$. Assuming knowledge of the valuation distribution, the seller sets prices $\p = (p_1, p_2, \ldots, p_n)$ for the $n$ items, where, unless otherwise stated, we will assume the items are indexed so that $p_1 \le p_2 \le \ldots \le p_n$.

Facing item prices $\p$, we suppose customers are unit-demand and min-buying, and thus purchase the minimum priced item such that their valuation is at least the price, or nothing if all prices exceed their valuations. Since items are indexed to be in ascending price order we can equivalently say a customer with valuation vector $\textbf{v} = (v_1, v_2, \ldots, v_n)$ purchases the item $i^* := \min\{i : v_i \ge p_i\}$, or no item if $v_i < p_i$ for all $i \in [n]$. We will denote the expected revenue of a seller using prices $\p$ by $\rev_{DP}(\p)$, and let $\rev_{DP} := \max_{\p}\rev_{DP}(\p)$ with the argument suppressed to denote the maximum achievable expected revenue. When $\p$ is such that the prices for the items are all the same, i.e., $\p = p\cdot \e_n$ for some $p \in \mathbb{R}$, where $\e_n := (1, 1, \ldots, 1) \in \mathbb{R}^n$ is the vector of 1's, we will refer to this as \textbf{u}niform \textbf{p}ricing and denote the revenue by $\rev_{UP}(p)$. When referring to the expected revenue of the optimal uniform pricing we again omit the argument and denote it by $\rev_{UP} := \max_p \rev_{UP}(p)$. We assume that all pricing problems considered have finite, attained optima and that $\rev_{UP} > 0$.

In this paper, we will compare $\rev_{DP}$ and $\rev_{UP}$ for several important classes of valuation distributions. In the remainder of this section, we will describe some preliminary results about these revenue functions that will aid us in our subsequent analysis.

\subsection{Preliminaries} \label{sec:prelims}

In this subsection, we will describe some preliminary results related to the structure of unit-demand pricing for min-buying customers. First, note that when valuations are independent, the expected revenue of a uniform price $p$ can be expressed simply as: 
\
\begin{equation}\label{eq:rev_up}
    \rev_{UP}(p) = p\left(1 - \prod_i \F_i(p)\right),
\end{equation}
\
which is the price $p$ multiplied by the probability that at least one valuation is at least that price. The revenue of uniform pricing is the same regardless of whether customers are min-buying or quasi-linear utility maximizing. 

The min-buying behavior becomes relevant for non-uniform price vectors. Specifically, for some prices $\p$, when valuations are independent and items are re-indexed so that $p_i \le p_{i+1}$ for all $i \in [n-1]$, the min-buying assumption allows us to express the revenue generated by $\p$ in a simple product/sum form:
\
\begin{align}
    \rev_{DP}(\p) :=  \sum_{i = 1}^n p_i \Pr \Bigl(i = \min\{j \in [n] : V_j \geq p_j\} \Bigr) 
    &= \sum_{i =1}^n  p_i \Pr( V_i \geq p_i \text{ and } V_j < p_j, \ \forall j \in [i-1] ) \notag\\
    &= \sum_{i =1}^n  p_i \Pr(V_i \geq p_i )   \prod_{j =1}^{i-1} \Pr(V_j < p_j) \notag\\
    &= \sum_{i =1}^n  p_i \cF_i(p_i) \prod_{j =1}^{i-1}  \F_j(p_j). \label{eq:rev_dp}
\end{align}
\
Often we will reference how a pricing policy operates over a restricted subset of items $S \subset [n]$. We denote the sub-vector of a price vector $\p$ over an index set $S$ by the subscript $\p_S$, i.e., $\p_S = (p_j)_{j \in S}$, and denote the revenue of a price vector over the restricted set $S$ by the superscript as: 
\
\begin{equation}\label{eq:ea-refined-dp}
\rev_{DP}^S(\p_S) := \sum_{i \in S} p_i \Pr\left(i = \min\{j \in S : V_j \geq p_j\}\right).
\end{equation}
\
For unit-demand pricing with min-buying customers and independent valuations, there is a natural recursive formula for the expected revenue of an item pricing relative to a restricted set of the items. Specifically, for any $k \in [n]$, let $S_k := \{1,\ldots,k\}$ be the set of the first $k$ items, and define $\bar S_k := [n] \setminus S_k$ to be the complement set of the last $n - k$ items. Then for any $k$, \cref{eq:rev_dp} can alternatively be expressed over $S_k$ and $\overline{S}_k$ as:
\
\begin{align}
    \rev_{DP}(\p) & = \sum_{i = 1}^k p_i \cF_{i}(p_i) \prod_{j =1}^{i-1}  \F_j(p_j) + \sum_{i = k+1}^n  p_i \cF_{i}(p_i) \prod_{j =1}^{i-1}  \F_j(p_j) \notag \\
    & = \rev_{DP}^{S_k}(\p_{S_k}) + \left(\sum_{i = k+1}^n  p_i \cF_{i}(p_i) \prod_{j = k+1}^{i-1}  \F_j(p_j) \right) \prod_{j = 1}^{k} \F_j(p_j) \notag \\ 
    & = \rev_{DP}^{S_k}(\p_{S_k}) + \rev_{DP}^{\bar S_k}(\p_{\overline{S}_k})  \prod_{j = 1}^{k} \F_j(p_j). \label{eq:rec_rev}
\end{align}

\cref{eq:rec_rev} captures a form of separability in revenue between the $k$ lowest priced items and the remaining $n-k$. If applying this property to $S_1$ we have:
\
\begin{align}\label{eq:obj_rec_1} 
\rev_{DP}(\p) 
& = \rev_{DP}^{S_1}(p_{S_1}) + \rev_{DP}^{\bar S_1}(\p_{\overline{S}_1})\F_1(p_{S_1}),
\end{align}
which neatly separates the $n$ item pricing problem for min-buying customers into a single item pricing problem in the first term (the term $\rev_{DP}^{S_1}(p_{S_1})$ equals $p_{S_1} \cF_1(p_{S_1})$ which is the revenue of item 1 alone offered at the scalar price $p_{S_1}$), and an $n-1$ item pricing problem in the second term. The two terms are tied together only by the constraint that the prices are ascending. Specifically, from \cref{eq:obj_rec_1} we observe that if $\p^*$ is an optimal price vector, then $\p^*_{\overline{S}_1}$ can be chosen to be an optimal item pricing over $\overline{S}_1$ subject to the constraint that the prices are increasing, i.e., that $\p_{\overline{S}_1} \ge p^*_{S_1} \e_{n-1}$ component-wise.

This is a key way in which the revenue function of item pricing for min-buying customers differs from item pricing for quasi-linear utility maximizers. In the quasi-linear case, items ``compete'' through their net utilities $v_j - p_j$, so the demand for any one item depends on the full price vector. In contrast, min-buying induces a strong form of cannibalization where lower-priced items screen customers away from higher-priced items. Armed with this insight, we are now ready to study UDP for min-buying customers when valuations are symmetric. 

\section{Uniform Pricing vs. Discriminatory Pricing}\label{sec:symm}
In this section, we will study unit-demand pricing for min-buying customers when valuations are drawn symmetrically. This setting models discriminatory pricing where a seller offers a substitutable set of products, for instance various colors of the same item, at different prices. As mentioned in the introduction, a surprising and influential line of work has shown that, even under symmetric valuations, discriminatory pricing can substantially outperform uniform pricing when customers are quasi-linear utility maximizers. Here, we consider when this advantage persists for min-buying customers by studying the performance of uniform pricing in three increasingly general symmetric environments: IID valuations, exchangeable and affiliated valuations, and general exchangeable valuations.

\subsection{Uniform Pricing Guarantees for IID Valuations}\label{sec:iid}
The case of IID valuations is the most commonly studied symmetric setting of the unit-demand pricing problem. For quasi-linear utility-maximizing customers with IID valuations, \cite{chawla2007algorithmic,dutting2016revenue} show that discriminatory pricing can earn up to twice as much revenue as uniform pricing, and that this factor is tight as the number of items grows large. Our first result shows that this phenomenon disappears for min-buying customers: when valuations are IID, discriminatory pricing offers no advantage over uniform pricing.

\begin{theorem}[IID Valuations]\label{thm:iid}
    Suppose there are $n$ items, customers are min-buying, and they draw valuations IID from some positive-valued distribution $V_i \sim \F$ for $i \in [n]$. Then a uniform price is revenue-maximizing, i.e., $\rev_{DP} = \rev_{UP}$.
\end{theorem}

\cref{thm:iid} shows that, for min-buying customers and IID valuations, the natural intuition that symmetrically valued products should be priced symmetrically holds true. \cref{thm:iid} also implies that the algorithmic problem of setting prices in this setting is straightforward. Whereas, in the case of quasi-linear utility-maximizing customers, the unit-demand pricing problem is NP-hard to solve even when valuations are IID \citep{chen2018complexity}, when customers are min-buying the corresponding algorithmic problem simplifies to a one-dimensional optimization to determine the optimal uniform price.

Before we formally prove \cref{thm:iid}, to build intuition it is helpful to consider how non-uniform prices can capture revenue in symmetric settings. Specifically, a non-uniform price vector creates high- and low-price ``tiers''. The high-price tiers can increase revenue in cases where a customer both has strong preferences and those preferences happen to be for high-priced items. This benefit is offset by a cannibalization effect where the low-priced items can draw demand away from higher-priced items. For min-buying customers, this cannibalization effect is much stronger than in the quasi-linear case. A min-buying customer who finds any lower-priced item acceptable (i.e. would get non-negative utility from purchase) never purchases a higher-priced acceptable item, regardless of how much more they value the higher-priced item. As discriminatory pricing yields no advantage for min-buying customers, this implies that revenue lift for discriminatory pricing in the quasi-linear case is driven by customers who find multiple items acceptable and nevertheless opt for the more expensive of the acceptable items. We prove this formally next. 

% \subsection{Proof of \cref{thm:iid}}

\begin{proof}{Proof of \cref{thm:iid}.}
We will prove that revenue is maximized by uniform pricing by induction on the number of items, $n$. In fact, we will prove the stronger statement that, for every price vector $\bp = (p_1, \ldots, p_n)$, there exists a price $\hat p \in \{p_1, \ldots, p_n\}$ such that $\rev_{UP}(\hat p) \ge \rev_{DP}(p_1, \ldots, p_n)$. The base case of $n = 1$ follows immediately, so let $n \ge 2$ and suppose the claim holds for $n-1$ items. Fix an arbitrary price vector $\p = (p_1, \ldots, p_n)$ and suppose the items are re-ordered so that $p_1 \le \ldots \le p_n$ (by symmetry this is WLOG). 

It follows from \cref{eq:obj_rec_1} and the induction hypothesis that there exists a price $\hat p \in \{p_2, \ldots, p_n\}$ (and thus $\hat p \geq p_1$), such that:
\
\begin{align}\label{eq:iid_reduce_two_tiers}
    \rev_{DP}(\p)
    &= p_1 \bigl(1 - \F(p_1)\bigr) + \F(p_1) \rev_{DP}^{\overline S_1}(p_2, \ldots, p_n) \notag\\
    &\le p_1\bigl(1 - \F(p_1)\bigr) + \F(p_1) \rev_{UP}^{\overline S_1}(\hat p) \notag\\
    &= \rev_{DP}\bigl(p_1,\hat p \e_{n-1}\bigr),
\end{align}
i.e., that the revenue $\rev_{DP}(\p)$ is dominated by a discriminatory pricing $\tilde \p := (p_1, \hat p, \ldots, \hat p)$. 

Now we will show $\rev_{DP}\bigl(p_1,\hat p \e_{n-1}\bigr)$ is dominated by one of two possible uniform prices, either $p_1$ or $\hat p$. To that end, first note that if $\F(\hat p) = 1$, then
\
\begin{equation*}
    \rev_{DP}\bigl(p_1,\hat p\e_{n-1}\bigr) 
    = p_1\bigl(1 - \F(p_1)\bigr)
    \le \rev_{UP}(p_1),
\end{equation*}
and the result follows. Thus suppose that $\F(\hat p) < 1$ and note since $p_1 \le \hat p$, this also implies that $\F(p_1) < 1$. Define $\lambda := \frac{1 - \F(p_1)}{1 - \F(p_1)^n} $, and consider:
\
\begin{align}\label{eq:iid_convex_combination}
    \max\left\{\rev_{UP}(p_1), \rev_{UP}(\hat p) \right\}
    &\ge \lambda \rev_{UP}(p_1) + (1 - \lambda) \rev_{UP}(\hat p) \notag\\
    &= p_1 \bigl(1 - \F(p_1)\bigr) + \hat p \F(p_1) \left(1-\F(\hat p)^{n-1}\right) 
    \frac{1-\F(\hat p)^n}{1-\F(\hat p)^{n-1}} \frac{1-\F(p_1)^{n-1}}{1-\F(p_1)^n} \notag\\
    & \geq p_1 (1 - \F(p_1)) +  \F(p_1) \hat p (1 - \F(\hat p)^{n-1})  \nonumber \\
    & = \rev_{DP}(\tilde{\p}), 
\end{align}
where the first inequality follows from the fact that $0 < \lambda \leq 1$, and the second inequality follows because (i) the function $\frac{1-x^n}{1-x^{n-1}}$ is increasing for $x \in [0,1)$, and (ii) since $p_1 \leq \hat p$ we have that $\F(p_1) \le \F(\hat p)$ and thus,
\
\begin{equation*}
    \frac{1 - \F(\hat p)^n}{1 - \F(\hat p)^{n-1}} \geq \frac{1 - \F(p_1)^n}{1 - \F(p_1)^{n-1}}.
\end{equation*}
Finally, the last equality follows from \cref{eq:obj_rec_1}. Thus from \cref{eq:iid_reduce_two_tiers} and \cref{eq:iid_convex_combination}, we see that uniform pricing always dominates any discriminatory pricing. By the principle of mathematical induction, uniform pricing is revenue maximizing for all $n$. 
\hfill\Halmos\end{proof}

We note that the proof of \cref{thm:iid} requires no regularity assumptions on the valuation distribution, nor does it require a detour through mechanism design or prophet inequalities as optimality proofs often do in the quasi-linear setting. In this sense, it is a purely pricing-based proof that proceeds by exploiting the special structure generated by min-buying behavior and IID valuations. The proof of \cref{thm:iid} also serves as a useful template for more general proofs of uniform pricing optimality: proceeding by induction on the number of items, reducing to the two price case, and then showing that discriminatory pricing is upper-bounded by a convex combination of two uniform-pricing revenues. We will extend this argument in the next subsection. 

\subsection{Uniform Pricing Guarantees for Exchangeable and Affiliated Valuations}\label{sec:affiliated}
Although IID valuations are commonly studied, they are not the most general model of symmetric valuations. In the remainder of this section, we consider the broader symmetric setting in which valuations are drawn from an exchangeable distribution. Exchangeability requires that the joint distribution of the valuation vectors be invariant to any permutation of the items. Thus, while exchangeability preserves symmetry across items, it still allows for item valuations to be correlated. This correlation makes exchangeable distributions expressive enough to capture a variety of market characteristics that IID valuations cannot. For example, an exchangeable distribution can represent variation in breadth of preference by controlling how many items a customer is likely to value highly. It can also capture heterogeneity in the strength of preferences by allowing for mixtures of high- and low-valuation customer types (as we do in \cref{ex:affiliated_ql_gap}). 

Unfortunately, uniform pricing is not optimal for all exchangeable distributions, even when customers are min-buying. To illustrate, we give a simple example where discriminatory pricing outperforms uniform pricing in \cref{ex:simple_exchangeable_gap} next. 

\example[$\rev_{DP}>\rev_{UP}$ for an Exchangeable Distribution]
\label{ex:simple_exchangeable_gap}
Suppose there are two items, customers are min-buying, and their valuations are drawn according to the exchangeable three-point distribution:
\[
    (V_1, V_2)
    =
    \begin{cases}
        (1,1), & \text{with probability } \frac12,\\
        (2,0), & \text{with probability } \frac14,\\
        (0,2), & \text{with probability } \frac14.
    \end{cases}
\]
For this instance, a uniform price of 1 and uniform price of 2 both earn revenue 1 and thus $\rev_{UP} = 1$. Now consider the discriminatory pricing $\p = (1, 2)$. Customers with a valuation vector of $(1, 1)$ purchase the first item at price $1$, customers with a valuation vector of $(2,0)$ also purchase the first item at price $1$, and customers with a valuation vector $(0,2)$ purchase the second item at price $2$. Therefore,
\[
    \rev_{DP}(1,2) = \frac12 \cdot 1 + \frac14 \cdot 1 + \frac14 \cdot 2
    = \frac54.
\]
Thus, discriminatory pricing earns strictly more revenue than uniform pricing, even with only two items and an exchangeable distribution supported on three valuation vectors.
\hfill\Halmos\endexample

\cref{ex:simple_exchangeable_gap} presents a barrier to generalizing \cref{thm:iid} to all exchangeable valuation distributions. However, note that the distribution in \cref{ex:simple_exchangeable_gap} encodes a form of negative dependence among item valuations: conditional on a customer valuing one item highly, at $2$, their valuation for the other item is $0$. Although such negative dependence may be plausible in some settings, when items are valued symmetrically because they are substitutes within a common product category, it may also be natural to suppose that valuations are positively associated. That is, customers who value some items highly may also tend to value the remaining items highly. In the mechanism-design literature, this form of positive dependence is commonly modeled through an \textit{affiliation} condition \citep{milgrom1981good,milgrom1982theory,anunrojwong2022robustness}. We give a formal definition next.

\begin{definition}[Affiliation]\label{ass:affiliation}
A joint distribution over valuations $V_1, \ldots, V_n$ is \textbf{affiliated} if it has a joint density or probability mass function $f$ such that, for all valuation vectors $\textbf{v}$ and $\textbf{v}'$,
\
\begin{equation}\label{eq:ea-refined-affiliation}
    f(\textbf{v} \vee \textbf{v}') f(\textbf{v} \wedge \textbf{v}') \ge f(\textbf{v}) f(\textbf{v}'),
\end{equation}
%-
where $\textbf{v} \vee \textbf{v}'$ and $\textbf{v} \wedge \textbf{v}'$ denote the component-wise maximum and minimum of $\textbf{v}$ and $\textbf{v}'$, respectively.
\end{definition}

Unpacking \cref{eq:ea-refined-affiliation}, it encodes positive dependence by enforcing that aligned valuations where the vector is all high (the pairwise maxes) or all low (the pairwise minimums) are more likely than mixes of high and low valuations.  This is easiest to see for two items, where affiliation is equivalent to requiring that $f(x_1, x_2) f(y_1, y_2) \ge f(x_1, y_2) f(y_1, x_2)$ for all $x_1 < y_1$ and $x_2 < y_2$. 

As IID valuations are trivially affiliated, the setting of exchangeable and affiliated valuations is strictly more general than the IID setting studied in \cref{thm:iid}. Moreover, the distribution in \cref{ex:simple_exchangeable_gap} is not affiliated. Thus we may consider whether uniform pricing remains optimal when valuations are both exchangeable and affiliated. The challenge in handling this more general setting is that the positive dependence among item valuations can be quite intricate. Furthermore, for quasi-linear customers, there exist exchangeable and affiliated distributions under which discriminatory pricing earns strictly more than twice the revenue of the optimal uniform pricing; we provide a simple construction in \cref{ex:affiliated_ql_gap}. Thus, relative to the IID setting, discriminatory pricing can be even more powerful for quasi-linear customers when valuations are exchangeable and affiliated. Nevertheless, imposing affiliation in addition to exchangeability is sufficient to recover the optimality of uniform pricing for min-buying customers, as we show next.

\begin{theorem}[Exchangeable and Affiliated Valuations]\label{thm:exchang_aff}
Suppose there are $n$ items, customers are min-buying, and their valuations are drawn according to an exchangeable and affiliated joint distribution. Then a uniform price is revenue-maximizing, i.e., $\rev_{DP} = \rev_{UP}$. 
\end{theorem}

The proof of \cref{thm:exchang_aff}, while more involved, follows the same general structure as the proof of \cref{thm:iid}. First, we use the induction hypothesis, together with standard closure properties of exchangeable and affiliated distributions, to reduce from considering arbitrary discriminatory price vectors to just two-level pricings in which the first item has a low price ($p_1$) and all remaining items have a common higher price ($\hat p > p_1$). We then bound the revenue of this two-level pricing by the maximum of $\rev_{UP}(p_1)$ and $\rev_{UP}(\hat p)$. However, an additional issue arises; under the two-level pricing, either the first item is purchased at price $p_1$, or it is rejected, meaning $V_1 < p_1$, and the remaining items are considered for purchase. Unlike in the IID setting, for exchangeable distributions rejection of the first item may change the joint distribution of the remaining valuations, so we must control how this conditioning affects the probability of a future purchase. Intuitively, it is here where affiliation rescues the proof as it implies that customers whose valuation for the first item is low also tend to have lower valuations for the higher-priced items, lowering the revenue potential of discriminatory pricing. 

Formally, by following the same convex-combination argument used in the proof of \cref{thm:iid}, we show that to gain the necessary control over the revenue after conditioning in the exchangeable and affiliated case, it is enough to establish the following inequality. Letting $M := \max_i V_i$, for all $a < b$ we require that
\begin{equation}\label{eq:max_rejection_main}
    \Pr\left(V_1 < a \mid M \ge a \right) \ge \Pr\left(V_1 < a \mid M \ge b\right).
\end{equation}
In words, \cref{eq:max_rejection_main} states that conditional on the customer's maximum valuation exceeding a higher threshold, any particular valuation is less likely to remain below the low-price threshold $a$. Unfortunately, we can not prove \cref{eq:max_rejection_main} by applying the affiliation condition (i.e. \cref{eq:ea-refined-affiliation}) directly to the pair $(V_1, M)$. The maximum $M$ is an aggregation of the valuation vector which includes $V_1$, and affiliation is not generally preserved under arbitrary increasing aggregations. Positive association alone is also insufficient: it compares fixed increasing events, whereas \cref{eq:max_rejection_main} requires a monotone comparison across different conditioning thresholds.

We resolve this difficulty by passing to the full vector of order statistics. Let $V^{(k)}$ denote the $k^{th}$ order statistic and note $M = V^{(n)}$. By exchangeability, the order-statistic vector has density $g(y_1, \ldots, y_n) := n!f(y_1,\ldots,y_n) \mathbf 1\{y_1\le\cdots\le y_n\}$, and importantly, the ordered cone $\{y_1 \le \ldots \le y_n\}$ is a sub-lattice. Its indicator is therefore affiliated, and multiplying the original affiliated density by this indicator preserves affiliation \citep{karlin1981total}. Hence the order statistics are affiliated. Marginalizing the order-statistic vector then shows that each pair $(V^{(k)}, M)$, for $k < n$ is affiliated. This affiliation then implies that, for every fixed $a$ and $k < n$, the function $\phi_{k, a}(t) := \Pr(V^{(k)} < a \mid M \ge t)$ is non-increasing. Finally, using exchangeability and an averaging argument, we translate these order-statistic comparisons back to an individual coordinate to show that, conditional on the symmetric event $\{M \ge t\}$, the probability that $V_1 < a$ equals the average probability that each order statistic lies below $a$. This establishes \cref{eq:max_rejection_main} and completes the comparison between the two-level pricing and the two corresponding uniform prices.

In the next subsection, we turn to general exchangeable distributions and show that, without positive-dependence assumptions, the worst-case performance of uniform pricing deteriorates to $\Theta(\log n)$ under both min-buying and quasi-linear utility maximization.

\subsection{Uniform Pricing for General Exchangeable Distributions}\label{sec:exch}
In this final subsection, we complement our results demonstrating the optimality of uniform pricing for min-buying customers in the IID and affiliated and exchangeable settings with a negative result showing that, in the most general symmetric settings, even under min-buying, uniform pricing may earn logarithmically less than optimal discriminatory pricing. For general exchangeable valuations and quasi-linear utility-maximizing customers, \cite{elmachtoub2021power} show that uniform pricing is a $1 + \log(n)$ approximation to the revenue of optimal discriminatory pricing, and this guarantee applies to min-buying customers as well since the revenue of any pricing with quasi-linear customers upper-bounds the revenue with min-buying ones (for uniform pricing they are identical). \cite{elmachtoub2021power} also conjecture that their bound is not tight and that a constant-factor approximation may be possible. In the following theorem, we show that this conjecture is incorrect and that, in fact, for both the min-buying and quasi-linear settings, the logarithmic approximation is best possible up to constant factors.

\begin{theorem}[Exchangeable Valuations]\label{thm:exchang}
Suppose there are $n$ items and customers are min-buying. There exists an exchangeable distribution such that $\rev_{DP} \ge \frac{3}{32} \log(n) \rev_{UP}$.
\end{theorem}

The proof of \cref{thm:exchang} follows by explicitly constructing an exchangeable distribution over multiple customer types, where each type $i$ is associated with a valuation level, $v_i$, and a number of items, $k_i \in [n]$, that the customer values. Conditional on a type, a subset of items of the prescribed size is chosen uniformly at random, and the customer assigns the same positive valuation to every item in that subset, valuing all remaining items as zero. As the valued subset is chosen uniformly at random for each customer, the resulting joint distribution is exchangeable. The valuation levels and type probabilities are then chosen in an equal-revenue fashion so that no single uniform price can extract more than a constant amount of revenue. 

The key feature of the construction is that customers with higher valuation levels value fewer items, creating a strong form of negative dependence in the valuations. Consequently, under a discriminatory price vector with geometrically decreasing prices, a type with a high valuation has a non-negligible probability of being directed toward an appropriately high-priced item, whereas types with lower valuations are more likely to purchase at one of the lower prices. Each valuation level therefore contributes on the order of the reciprocal of the number of remaining price levels to discriminatory-pricing revenue. Summing these contributions yields a harmonic, and hence logarithmic, revenue advantage.

Combined with Theorem~3 of \citet{elmachtoub2021power}, \cref{thm:exchang} implies that, for general exchangeable valuations, the worst-case approximation factor of uniform pricing is $\Theta(\log n)$ for both min-buying and quasi-linear utility-maximizing customers. Thus overall, uniform pricing under min-buying provides much stronger guarantees than in the quasi-linear case when valuations across items are symmetric and exhibit a form of positive association (affiliation). However, when the items are symmetric but their valuations may be arbitrarily negatively correlated, uniform pricing can perform similarly poorly under the two decision rules. 

In the next section, we move beyond symmetric environments and study the performance of uniform pricing when valuations across items are independent but not necessarily identically distributed across items.
\section{Uniform Pricing vs. Differentiated Pricing}\label{sec:reg}
In this section, we turn our attention to unit-demand pricing problems in which item valuations are independent but not identically distributed. These problems model differentiated pricing, where a seller offers an array of items that need not be perfect substitutes. For instance, a seller may offer standard, professional, and premium versions of software, where customers purchase at most one option but differences in their valuations can arise from differences in features, rather than solely from idiosyncratic preferences. We refer to pricing in this setting as \textit{differentiated}, rather than discriminatory, because differences in prices may be justified by genuine differences in item quality reflected in the valuations. With independent but non-identical valuations, we no longer expect uniform pricing to be optimal. Nevertheless, a surprising line of work shows that, when customers are quasi-linear utility maximizers and the valuation distributions are regular, a uniform price can achieve a significant fraction of the optimal differentiated-pricing revenue \citep{alaei2019optimal,jin2019tight}. Here, we study the analogous comparison between uniform and differentiated pricing for min-buying customers.

\subsection{Independent Regular Valuations}\label{sec:independent-regular}
Before stating our main result, we note that, with independent but non-identical valuations, some distributional restriction is necessary for uniform pricing to provide a nontrivial revenue guarantee. To see this, consider the general case in which valuations for the $n$ items are drawn independently but not identically according to $V_i \sim F_i$, $i \in [n]$. In this setting, an optimal uniform price always earns at least a $1/n$ fraction of the revenue of the optimal differentiated pricing, since uniform pricing can replicate the largest revenue contribution obtained from any one of the $n$ items. Unfortunately, without additional assumptions, this trivial guarantee is tight even when customers are min-buying. We demonstrate this in \cref{ex:notReg} in the appendix using a construction similar to one in \citet{hartline2009simple}. As no constant-factor approximation to differentiated pricing is possible in the general independent setting, it is standard to assume that the valuation distributions are regular. Formally, for a distribution $F$, let $\rev_F(q) := q F^{-1}(1 - q)$ for $q \in (0,1)$, with endpoint values defined by one-sided limits, denote $F$'s revenue-quantile function, where $F^{-1}$ is the generalized quantile function. We call $F$ \emph{regular} if $\rev_F$ is concave. Many common distributions are regular, including the uniform and exponential distributions. Regularity is also implied by several well-known stronger conditions, including the monotone hazard rate (MHR) condition and log-concavity of the density.

For the remainder of this section, we assume that each item's valuation is drawn independently from a regular distribution. Our main result characterizes the worst-case guarantee for the revenue of uniform pricing relative to optimal differentiated pricing. To state our result precisely, we first need to define two functions. Let $A(x) := -\frac{\log(1 - x^2)}{x}$ for $x \in (0, 1)$, with $A(0) := 0$, and let $c(x)$ solve the ordinary differential equation:
\begin{equation}\label{eq:reg-refined-ode}
    c'(x) = A'(x)\left(1 - x c(x)\right),
\end{equation}
with initial condition $c(0) = 1$. Let $x^\star \approx 0.6986$ be the unique solution to $x c(x) = 1$ (we prove uniqueness in \cref{sec:proof_reg}). We will show that, when valuations are independent and regular, the optimal uniform price guarantees approximately $69.86\%$ of the optimal differentiated-pricing revenue.

\begin{theorem}[Independent and Regular Valuations]\label{thm:indep_reg}
Suppose customers are min-buying and their valuations are drawn independently from continuous regular distributions $V_i\sim F_i$, $i\in[n]$. Then the optimal uniform- and differentiated-pricing revenues satisfy
\[
    \frac{\rev_{UP}}{\rev_{DP}}\ge x^\star \approx 0.6986,
\]
and this factor is tight.
\end{theorem}

\cref{thm:indep_reg} is the min-buying analogue of the tight result of \citet{jin2019tight} for quasi-linear customers. Through the equivalence between min-buying UDPs and order-oblivious posted pricing, it also characterizes the tight approximation ratio between anonymous pricing and optimal order-oblivious posted-price mechanisms. Note, \cref{thm:indep_reg} provides a substantially stronger guarantee for uniform pricing under min-buying than in the quasi-linear setting, increasing the guaranteed fraction of optimal differentiated-pricing revenue from approximately $38\%$ to nearly $70\%$. Thus, although uniform pricing is not generally optimal in asymmetric valuation environments, it is a remarkably effective and often near-optimal policy for the broad class of independent regular valuations. In the remainder of this section, we prove \cref{thm:indep_reg}.

\subsection{Proof of \cref{thm:indep_reg}}\label{sec:proof_reg}
At a high level, our proof follows an approach similar to that developed by \citet{alaei2019optimal,jin2019tight,jin2020tight} for the quasi-linear case, successively narrowing the class of valuation distributions we must consider until the search for a worst-case instance reduces to a one-dimensional continuous optimization. Specifically, the proof proceeds in four steps. First, we reduce from general regular distributions to so-called triangular-revenue distributions, whose revenue-quantile functions are triangular (cf. \cref{lem:reg-refined-ratio-reduction}). This reduction preserves the relevant differentiated-pricing revenue while weakly reducing uniform-pricing revenue. Second, we further reduce to a subclass of canonical triangular-revenue distributions in which the high value items are condensed into one item, the low value items are dropped, and where it suffices to consider differentiated pricing that prices each item at its ``cap'' (cf. \cref{lem:reg-refined-canonicalization}). Third, we associate the differentiated pricing revenue over canonical instances with a continuous upper-bound which we analyze to show that $\rev_{UP}/\rev_{DP}$ is bounded below by $x^\star$ (cf. \cref{lem:reg-refined-Tri-bound}). Finally, we construct a sequence of instances for which this ratio converges to $x^\star$, establishing the tightness of \cref{thm:indep_reg} (cf. \cref{lem:lower_bound}).

\subsubsection{Reduction to Triangular-Revenue Distributions}
\label{sec:reg-refined-reduction}
To begin, we identify the regular distributions that are least favorable to uniform pricing relative to differentiated pricing. For parameters $a, \v > 0$, we say that a random variable $V$ follows a \textit{triangular-revenue distribution}, denoted by $V \sim \Tri(a,\v)$, if its cumulative distribution function (CDF) satisfies:
\begin{equation}\label{eq:reg-refined-Tri-cdf}
    \Pr(V\le t)=\frac{t}{t+a}, \qquad 0\le t < \v,
\end{equation}
with the remaining probability mass placed at $\v$, i.e., $\Pr(V = \v) = a/(a + \v)$. As $\v$ is the maximum valuation, we refer to it as the \textit{cap} of the triangular-revenue distribution. The revenue-quantile function of this distribution is triangular:
\begin{equation}\label{eq:reg-refined-Tri-revenue-curve}
    \rev_{\Tri(a,\v)}(q)=
    \begin{cases}
        \v q, & 0\le q\le \frac{a}{a+\v},\\
        a(1-q), & \frac{a}{a+\v} \le q \le 1.
    \end{cases}
\end{equation}
The two branches of the revenue-quantile function meet at $q = a/(a + \v)$, where the slope changes from $\v$ to $-a$. Thus, the function is concave, and the item's revenue is maximized at price $\v$, corresponding to quantile $q = a/(a + \v)$. In the following lemma, we show that the worst-case analysis for independent regular valuations under min-buying can be reduced to triangular-revenue distributions.

\begin{lemma}[Reduction to Triangular-Revenue Distributions]
\label{lem:reg-refined-ratio-reduction}
Consider an instance in which customers are min-buying and their valuations are drawn independently from continuous regular distributions $V_i \sim F_i$, $i \in [n]$. There exists a collection of at most $n$ independent triangular-revenue distributions such that
\[
    \widetilde{\rev}_{DP} \ge \rev_{DP}
    \qquad\text{and}\qquad
    \widetilde{\rev}_{UP} \le \rev_{UP},
\]
where $\widetilde{\rev}_{DP}$ and $\widetilde{\rev}_{UP}$ are the optimal differentiated- and uniform-pricing revenues, respectively, over the triangular-revenue distributions.
\end{lemma}

The intuition behind \cref{lem:reg-refined-ratio-reduction} follows similar reductions in the quasi-linear literature. Specifically, once an item's revenue at some designated differentiated price $p_i$ is fixed, regularity forces a minimum amount of demand for that item at every lower price $p < p_{i}$. The idea of a triangular-revenue distribution reduction is that it can replace the original valuation distribution in such a way that it preserves the item's acceptance probability at $p_i$, i.e., the probability that $V_i \ge p_i$ and thus it could be considered for purchase, while making the item as unattractive as regularity permits at any lower price. As preserving these acceptance probabilities preserves the revenue of the overall differentiated price vector, and the reduction hurts every uniform pricing, triangular-revenue distributions are therefore extremal for the ratio of uniform to differentiated pricing.

Given \cref{lem:reg-refined-ratio-reduction}, we no longer need to consider arbitrary concave revenue-quantile functions. Instead, each item can be associated with a triangular-revenue distribution and characterized by only two parameters: its cap $\v$, which is the maximum price it could be offered for, and $a$, which determines its acceptance probability at prices below the cap. The next question is then, how large can the revenue gap between differentiated and uniform pricing be for independent triangular-revenue valuations?

\subsubsection{Further Reduction to Canonical Triangular-Revenue Distributions}
\label{sec:reg-refined-canonical}
\cref{lem:reg-refined-ratio-reduction} reduces the analysis of independent regular valuations to instances consisting of independent two-parameter triangular-revenue distributions. Although these instances are simpler, optimizing the worst-case ratio over arbitrary triangular-revenue instances remains challenging. In particular, the different items may yield revenue at wildly different scales, and an optimal differentiated pricing may set some prices below their corresponding caps, leaving an additional set of price variables to optimize over. In the next step, we further reduce the class of instances we must consider to one in which the relatively high-valued items are condensed and where it suffices to analyze pricing each item at its cap.

The reduction is driven by a simple property of min-buying. Define the \emph{continuation revenue} from a collection of (relatively) expensive items as the expected revenue they generate conditional on the customer reaching them, or equivalently, conditional on no cheaper item being purchased. Specifically, suppose the continuation revenue from all items priced above $\v$ is $C$, and consider adding an independent item $V_j \sim \Tri(a,\v)$ priced at its cap, $p_j = \v$. The new item is purchased with probability $a/(a + \v)$; otherwise, the customer rejects it and proceeds to the more expensive items. Hence, the new continuation revenue is
\begin{equation}\label{eq:continuation}
    T_{a, \v}(C)
    = \frac{a}{a + \v}\v + \frac{\v}{a + \v}C
    = \frac{a + C}{1 + a/\v}.
\end{equation}
Thus, once the items are ordered by price, the entire collection of more expensive items affects the revenue calculation for each cheaper item only through the scalar continuation value $C$. Once $C$ is fixed, the size/cardinality, individual distributions, and prices of the more expensive items need not be tracked separately. This scalar-state property of min-buying allows us to compress a complicated triangular-revenue instance into what we call a \textit{canonical} triangular-revenue instance, following the terminology of \citet{alaei2019optimal}. Our reduction successively removes information that is irrelevant to the continuation recursion while ensuring that uniform pricing does not become more profitable.

Formally, a canonical instance consists of a terminal item $V_0 \sim \Tri(1, P)$ and $N$ additional independent items $V_k \sim \Tri(a_k, \v_k)$, whose caps are indexed in decreasing order, $P > \v_1 > \ldots > \v_N > 1$. We then evaluate a canonical instance under its \textit{cap-pricing}, which is the pricing that sets each item's price at its cap with $p_0 = P$ and $p_k = \v_k$ for each $k \in [N]$. Let $C_k$ denote the continuation revenue obtained after starting with the terminal item and incorporating the first $k$ additional items in decreasing-cap order. By \cref{eq:continuation},
\[
    C_0 = \frac{P}{P + 1},
    \qquad
    C_k = \frac{a_k + C_{k-1}}{1 + a_k/\v_k},
    \qquad k \in [N].
\]
The importance of this representation is that, instead of optimizing over arbitrary triangular-revenue distributions and differentiated prices, we need only consider an ordered sequence of caps $(\v_1, \ldots, \v_N)$ and associated parameters $(a_1, \ldots, a_N)$, with the cap-pricing revenue governed by the scalar recursion in \cref{eq:continuation}.

The next lemma shows that every triangular-revenue instance can be transformed into a canonical instance without increasing optimal uniform-pricing revenue and with an arbitrarily small loss, if any, relative to the original optimal differentiated-pricing revenue.

\begin{lemma}[Reduction to Canonical Triangular-Revenue Distributions]
\label{lem:reg-refined-canonicalization}
Consider an instance $\cI$ of $n$ independent triangular-revenue valuations $V_i \sim F_i$, $i \in [n]$, normalized so that $\rev_{UP}=1$. For every $\ep > 0$, there exists an integer $N \in \{0, 1, \ldots, n\}$, a value $P > 1$, and a canonical instance $\widetilde{\cI}$ consisting of independent valuations $\widetilde V_0 \sim \Tri(1,P)$ and $\widetilde V_k \sim \Tri(a_k, \v_k)$, $k \in [N]$, with caps satisfying $P > \v_1 > \ldots > \v_N > 1$, such that, under the cap-pricing $\widetilde\p = (P, \v_1, \ldots, \v_N)$,
\[
    \widetilde{\rev}_{DP}(\widetilde\p) \ge \rev_{DP} - \ep
    \qquad \text{and} \qquad
    \widetilde{\rev}_{UP} \le 1,
\]
where $\widetilde{\rev}_{DP}(\widetilde\p)$ and $\widetilde{\rev}_{UP}$ are the revenues of the differentiated pricing $\widetilde\p$ and the optimal uniform-pricing on the transformed instance $\widetilde{\cI}$, respectively.
\end{lemma}

The proof of \cref{lem:reg-refined-canonicalization} consists of three conceptually simple transformations. First, if an item with a triangular-revenue distribution is priced below its cap, we can truncate the distribution at its designated differentiated price without changing its acceptance probability at that price, while weakly reducing its demand under uniform pricing. Second, items with a common cap can be aggregated into a single triangular-revenue item. Third, after indexing the remaining items in decreasing-cap order, we compress the first several items, those with the largest caps and hence the highest differentiated prices, into a terminal equal-revenue component and at most one boundary item. Similar aggregation and ordered-triangular reductions have been used by \citet{alaei2019optimal,jin2020tight}. Here, however, the transformations are chosen to preserve or improve the min-buying continuation revenue. The proof also explains why items with caps below one can be excluded from the worst-case analysis.

\cref{lem:reg-refined-canonicalization} preserves the two comparisons needed for the worst-case analysis: the canonical instance has uniform-pricing revenue no greater than the normalized benchmark, while its designated cap-pricing loses at most an arbitrarily small amount of revenue relative to the original differentiated-pricing optimum. It is therefore sufficient to bound the cap-pricing revenue of canonical triangular-revenue instances.

\subsubsection{Bound for Canonical Triangular-Revenue Distributions}
\label{sec:reg-refined-Tri-bound}
The reduction to canonical instances leaves a substantially simpler extremal problem. Starting from the terminal continuation value $C_0 = P/(P+1)$, which approaches one as $P$ becomes large, it remains only to characterize how large the recursion $C_k=(a_k + C_{k-1})/(1 + a_k/\v_k)$ can become while optimal uniform-pricing revenue is constrained to be at most one. To achieve this, we must account for the tradeoff between creating revenue opportunities at individual caps and increasing demand under uniform pricing. Consider an item $k$ with valuation distribution $\Tri(a_k, \v_k)$ priced at its cap $\v_k$. Increasing $a_k$ raises the probability that the item is accepted at $\v_k$, thereby increasing its standalone revenue at that price. However, the same increase also raises its acceptance probability $a_k/(a_k + p)$ at every lower uniform price $p \le \v_k$. Thus, creating a larger revenue opportunity at one price simultaneously increases demand under every lower uniform price. Thus, revenue opportunities at different prices cannot be created independently.

This coupling can be seen directly from the constraint that uniform-pricing revenue not exceed one. At the uniform price $p = \v_k$, only the first $k$ items and the terminal item can be purchased, since their caps are at least $\v_k$. Every remaining item has a cap below $\v_k$ and is therefore rejected with probability one. The terminal item $V_0 \sim \Tri(1, P)$ is rejected with probability $\v_k/(\v_k + 1)$, while item $i \in [k]$ is rejected with probability $\v_k/(\v_k + a_i)$. Hence, by independence,
\begin{equation}\label{eq:uni_const}
    \rev_{UP}(\v_k)
    = \v_k\left(1 - \frac{\v_k}{\v_k + 1} \prod_{i=1}^k \frac{\v_k}{\v_k + a_i}\right)
    \le 1.
\end{equation}
Let $x_k := 1/\v_k$ denote the reciprocal cap. Rearranging and taking logarithms of \cref{eq:uni_const} gives the equivalent constraint:
\[
    \sum_{i=1}^k \log(1 + a_i x_k) \le -\log(1 - x_k^2).
\]
Note each item's contribution on the left still depends on the reciprocal cap $x_k$ at which the constraint is evaluated. To separate the item parameters from this price dependence, we use concavity of $\log(1 + a_ix)$, which gives $\log(1 + a_ix_k) \ge x_k\log(1+a_i)$. Applying this inequality and dividing by $x_k > 0$ yields the relaxed constraint:
\[
    \sum_{i=1}^k \log(1 + a_i)
    \le -\frac{\log(1 - x_k^2)}{x_k} = A(x_k).
\]
We interpret $A(x_k)$ as the cumulative uniform-pricing budget available up to reciprocal price $x_k = 1/\v_k$, with each item contributing the price-independent quantity $\log(1 + a_i)$. For fixed caps, the parameters $a_1, \ldots, a_k$ therefore cannot be chosen independently: increasing one parameter may generate additional differentiated-pricing revenue at its own price, but it consumes part of the same budget and leaves less room for revenue opportunities at other prices. This cumulative constraint captures how regularity links revenue opportunities across prices and simplifies the subsequent analysis of the continuation recursion.

This perspective is closely related to normalized common-price programs in the anonymous-pricing literature. For instance, \citet{alaei2019optimal} normalize anonymous-pricing revenue to one and impose the corresponding constraint at every price, and \citet{jin2020tight} exploit the triangular structure to obtain tractable cumulative constraints at the relevant price boundaries. Our analysis builds on this perspective, but with a different objective: under min-buying, allocating triangular mass across price scales affects differentiated-pricing revenue through the continuation recursion $C_k = (a_k + C_{k-1})/(1 + a_k x_k)$. We must therefore bound this recursion subject to the cumulative budget constraints above.

The key step is to compare the discrete canonical problem with a continuous relaxation. We first move each reciprocal cap to the smallest value consistent with its cumulative budget, making the budget constraint bind at every step. This change can only increase continuation revenue. After upper-bounding the terminal continuation value by one, we show that the resulting relaxed recursion is bounded above by the continuous benchmark $c(x)$, defined in \cref{eq:reg-refined-ode} by
\[
    c'(x)=A'(x)\bigl(1 - x c(x)\bigr), \qquad c(0)=1.
\]
This differential equation is the infinitesimal counterpart of the discrete continuation recursion. We also prove that the benchmark attains its maximum at $x^\star$, where $c(x^\star) = \frac{1}{x^\star}$. Thus, no canonical instance with uniform-pricing revenue at most one can generate cap-pricing revenue greater than $1/x^\star$. The following lemma summarizes this bound.

\begin{lemma}[Bound for Canonical Triangular-Revenue Distributions]
\label{lem:reg-refined-Tri-bound}
Consider a canonical triangular-revenue instance with independent valuations $V_0 \sim \Tri(1, P)$ and $V_k \sim \Tri(a_k, \v_k)$ for $k \in [N]$, where $a_k > 0$ for each $k \in [N]$ and $P > \v_1 > \ldots > \v_N > 1$. If the optimal uniform-pricing revenue satisfies $\rev_{UP} \le 1$, then the revenue from pricing each item at its cap satisfies
\[
    \rev_{DP}(P, \v_1, \ldots, \v_N) \le \frac{1}{x^\star}.
\]
\end{lemma}

\cref{lem:reg-refined-Tri-bound} completes the upper-bound argument in \cref{thm:indep_reg}. Starting with $n$ independent and regular valuation distributions, first normalize the instance so that $\rev_{UP}=1$. By \cref{lem:reg-refined-ratio-reduction}, we may then replace it by a triangular-revenue instance without decreasing differentiated-pricing revenue or increasing uniform-pricing revenue. If the transformed instance has uniform-pricing revenue below one, rescale its valuations and prices to make this revenue again equal to one; this can only increase its differentiated-pricing revenue. Next, apply \cref{lem:reg-refined-canonicalization} to produce a canonical instance whose cap-pricing loses at most an arbitrarily small amount of revenue. Finally, by \cref{lem:reg-refined-Tri-bound}, for every $\ep>0$,
\[
    \rev_{DP}-\ep
    \le\widetilde{\rev}_{DP}(\widetilde\p)
    \le\frac{1}{x^\star}.
\]
Taking $\ep \to 0$ and undoing the normalization gives $\rev_{UP}\ge x^\star\rev_{DP}$. All that now remains is to show that the bound is tight.

\subsubsection{Tightness of \cref{thm:indep_reg}}
\label{sec:reg-refined-tightness}
We now show that the factor $x^\star$ is tight by constructing a sequence of instances that reverse engineers the relaxation used in \cref{lem:reg-refined-Tri-bound}. As in the lower-bound constructions of \citet{alaei2019optimal,jin2020tight}, our construction will approximate the continuous extremal solution by increasingly fine, finite triangular-revenue instances.

To begin, fix an endpoint $X \in (0, 1)$. We divide the available budget $[0, A(X)]$ into many small pieces and associate one triangular-revenue item with each piece. Each item's cap is chosen so that the corresponding uniform-pricing constraint from \cref{eq:uni_const} binds. As the number of items increases and the terminal cap grows, the continuation recursion approaches the continuous envelope. Choosing $X = x^\star$ therefore yields cap-pricing revenue arbitrarily close to $1/x^\star$. The following lemma formalizes the resulting lower bound.

\begin{lemma}\label{lem:lower_bound}
For every $\ep > 0$, there exists a canonical triangular-revenue instance $\cI$ such that $\rev_{UP} = 1$ and $\rev_{DP} > \frac{1}{x^\star} - \ep$.
\end{lemma}

\cref{lem:lower_bound} shows that the continuous-envelope bound is asymptotically attainable. The uniform-pricing constraints can be made to bind on an increasingly dense collection of prices, while the corresponding continuation recursion converges to the envelope. Thus, the constant $1/x^\star$ describes the worst-case revenue gap.

One final technical issue remains. Canonical triangular-revenue distributions contain atoms at their caps, whereas \cref{thm:indep_reg} is stated for continuous regular valuations. We address this issue in \cref{lem:continuous_approx} in the appendix by smoothing each triangular-revenue instance into an independent continuous regular instance. The smoothing preserves the revenue of the designated cap-pricing while changing optimal uniform-pricing revenue by an arbitrarily small amount. Applying \cref{lem:continuous_approx} to the instances in \cref{lem:lower_bound} shows that, for every $\ep > 0$, there exists an independent continuous regular instance satisfying $\frac{\rev_{DP}}{\rev_{UP}}>\frac{1}{x^\star}-\ep$. We defer the statement and proof of the smoothing result to \cref{sec:omitted_pfs_2}. Together with the upper bound, this establishes \cref{thm:indep_reg}.
\section{Conclusion}\label{sec:conclusions}
In this paper, we examined unit-demand pricing problems where customers are min-buying and established optimality and approximation results for uniform pricing policies. Our findings highlight the differences between pricing for min-buying customers and pricing for quasi-linear utility maximizers. Notably, when valuations are IID or, more generally, exchangeable and affiliated, we show that discriminatory pricing provides no additional benefit in min-buying models. Further, when valuations are independent and regular but not necessarily
identically distributed, we demonstrate that the advantages of differentiated pricing are substantially reduced for min-buying customers relative to the quasi-linear case. 

Our work thus highlights the role of customer choice assumptions in determining the revenue benefits of discriminatory and differentiated pricing. When valuations are IID for instance, our result implies that any additional revenue generated by discriminatory pricing over uniform pricing is driven by how customers choose among acceptable items. Quasi-linear utility-maximization is a standard and intuitive way to model this choice; however, it is only one of many choice rules. In settings where customers are boundedly rational or risk-averse, min-buying may also provide a useful model of purchasing behavior. Under this alternative choice rule, our results support using a well-chosen uniform price, establishing its optimality or providing approximation guarantees under a range of distributional assumptions.

Our work also provides tools for further research on pricing for min-buying customers. One direction is to extend our analysis to richer feasibility environments, including multi-unit and combinatorial settings, and to characterize when uniform or other simple pricing policies are optimal or provide strong guarantees for suitably risk-averse customers. The equivalence between min-buying UDPs and order-oblivious posted-price mechanisms also invites mechanism-design questions about the value of price discrimination and control over the order of offers, as well as the revenue cost of requiring anonymous prices. Finally, future work could study behavioral models that interpolate between min-buying and quasi-linear utility maximization, such as customers who choose the cheapest acceptable item with some probability and the utility-maximizing item otherwise. Related extensions could allow customers to compare only a limited subset of products. Such models could clarify how the revenue benefits of non-uniform pricing depend on the precise customer choice rule.

% For instance, we leave open the question of whether or not the result in \cref{thm:mix_iid} can be strengthen to show optimal for uniform pricing when valuations are exchangeable and affiliated as studied in \cref{sec:mixtures}. We conjecture that uniform pricing is indeed optimal in this setting (and have no counter-examples even after extensive search), however making further progress in this direction will require new ideas. Second, we leave open the problem pinning down the exact approximation factor for uniform pricing when customers are min-buying and valuations are drawn independently but non-identically from MHR distributions, as studied in \cref{thm:mhr}. We suspect that our current analysis is not tight, but advancing beyond it will similarly also likely require new techniques. 

\printendnotes

% \newpage
%-----------------------------------
% \ACKNOWLEDGMENT{}
%-----------------------------------

% \newpage
\bibliographystyle{ormsv080}
\bibliography{bibliography.bib}

@String{Computing = "Computing" }

@String{Springer = "Springer-Verlag" }

@inproceedings{guruswami2005profit,
  title={On profit-maximizing envy-free pricing.},
  author={Guruswami, Venkatesan and Hartline, Jason D and Karlin, Anna R and Kempe, David and Kenyon, Claire and McSherry, Frank},
  booktitle={SODA},
  volume={5},
  pages={1164--1173},
  year={2005}
}

@inproceedings{chawla2007algorithmic,
  title={Algorithmic pricing via virtual valuations},
  author={Chawla, Shuchi and Hartline, Jason D and Kleinberg, Robert},
  booktitle={Proceedings of the 8th ACM Conference on Electronic Commerce},
  pages={243--251},
  year={2007}
}

@inproceedings{chawla2010multi,
  title={Multi-parameter mechanism design and sequential posted pricing},
  author={Chawla, Shuchi and Hartline, Jason D and Malec, David L and Sivan, Balasubramanian},
  booktitle={Proceedings of the forty-second ACM symposium on Theory of computing},
  pages={311--320},
  year={2010}
}

@inproceedings{balcan2008item,
  title={Item pricing for revenue maximization},
  author={Balcan, Maria-Florina and Blum, Avrim and Mansour, Yishay},
  booktitle={Proceedings of the 9th ACM Conference on Electronic Commerce},
  pages={50--59},
  year={2008}
}

@article{elmachtoub2021power,
  title={The power of opaque products in pricing},
  author={Elmachtoub, Adam N and Hamilton, Michael L},
  journal={Management Science},
  volume={67},
  number={8},
  pages={4686--4702},
  year={2021},
  publisher={INFORMS}
}

@article{gallego2021bounds,
  title={Bounds and heuristics for multi-product personalized pricing},
  author={Gallego, Guillermo and Berbeglia, Gerardo},
  journal={arXiv preprint arXiv:2102.03038},
  year={2021}
}

@article{cai2015extreme,
  title={Extreme value theorems for optimal multidimensional pricing},
  author={Cai, Yang and Daskalakis, Constantinos},
  journal={Games and Economic Behavior},
  volume={92},
  pages={266--305},
  year={2015},
  publisher={Elsevier}
}

@inproceedings{chen2014complexity,
  title={The complexity of optimal multidimensional pricing},
  author={Chen, Xi and Diakonikolas, Ilias and Paparas, Dimitris and Sun, Xiaorui and Yannakakis, Mihalis},
  booktitle={Proceedings of the twenty-fifth annual ACM-SIAM symposium on Discrete algorithms},
  pages={1319--1328},
  year={2014},
  organization={SIAM}
}

@article{chen2018complexity,
  title={The complexity of optimal multidimensional pricing for a unit-demand buyer},
  author={Chen, Xi and Diakonikolas, Ilias and Paparas, Dimitris and Sun, Xiaorui and Yannakakis, Mihalis},
  journal={Games and Economic Behavior},
  volume={110},
  pages={139--164},
  year={2018},
  publisher={Elsevier}
}

@article{lucier2017economic,
  title={An economic view of prophet inequalities},
  author={Lucier, Brendan},
  journal={ACM SIGecom Exchanges},
  volume={16},
  number={1},
  pages={24--47},
  year={2017},
  publisher={ACM New York, NY, USA}
}

@article{rusmevichientong2006nonparametric,
  title={A nonparametric approach to multiproduct pricing},
  author={Rusmevichientong, Paat and Van Roy, Benjamin and Glynn, Peter W},
  journal={Operations Research},
  volume={54},
  number={1},
  pages={82--98},
  year={2006},
  publisher={INFORMS}
}

@article{correa2019pricing,
  title={From pricing to prophets, and back!},
  author={Correa, Jos{\'e} and Foncea, Patricio and Pizarro, Dana and Verdugo, Victor},
  journal={Operations Research Letters},
  volume={47},
  number={1},
  pages={25--29},
  year={2019},
  publisher={Elsevier}
}

@article{alaei2019optimal,
  title={Optimal auctions vs. anonymous pricing},
  author={Alaei, Saeed and Hartline, Jason and Niazadeh, Rad and Pountourakis, Emmanouil and Yuan, Yang},
  journal={Games and Economic Behavior},
  volume={118},
  pages={494--510},
  year={2019},
  publisher={Elsevier}
}

@article{jin2020tight,
  title={Tight revenue gaps among simple mechanisms},
  author={Jin, Yaonan and Lu, Pinyan and Tang, Zhihao Gavin and Xiao, Tao},
  journal={SIAM Journal on Computing},
  volume={49},
  number={5},
  pages={927--958},
  year={2020},
  publisher={SIAM}
}

@article{dutting2016revenue,
  title={Revenue gaps for discriminatory and anonymous sequential posted pricing},
  author={D{\"u}tting, Paul and Fischer, Felix A and Klimm, Max},
  journal={CoRR, abs/1607.07105},
  year={2016}
}

@inproceedings{jin2019tight,
  title={Tight approximation ratio of anonymous pricing},
  author={Jin, Yaonan and Lu, Pinyan and Qi, Qi and Tang, Zhihao Gavin and Xiao, Tao},
  booktitle={Proceedings of the 51st Annual ACM SIGACT Symposium on Theory of Computing},
  pages={674--685},
  year={2019}
}

@inproceedings{aggarwal2004algorithms,
  title={Algorithms for multi-product pricing},
  author={Aggarwal, Gagan and Feder, Tom{\'a}s and Motwani, Rajeev and Zhu, An},
  booktitle={Automata, Languages and Programming: 31st International Colloquium, ICALP 2004, Turku, Finland, July 12-16, 2004. Proceedings 31},
  pages={72--83},
  year={2004},
  organization={Springer}
}

@article{thanassoulis2004haggling,
  title={Haggling over substitutes},
  author={Thanassoulis, John},
  journal={Journal of Economic theory},
  volume={117},
  number={2},
  pages={217--245},
  year={2004},
  publisher={Elsevier}
}

@article{berbeglia2020assortment,
  title={Assortment optimisation under a general discrete choice model: A tight analysis of revenue-ordered assortments},
  author={Berbeglia, Gerardo and Joret, Gwena{\"e}l},
  journal={Algorithmica},
  volume={82},
  number={4},
  pages={681--720},
  year={2020},
  publisher={Springer}
}

@inproceedings{hartline2009simple,
  title={Simple versus optimal mechanisms},
  author={Hartline, Jason D and Roughgarden, Tim},
  booktitle={Proceedings of the 10th ACM conference on Electronic commerce},
  pages={225--234},
  year={2009}
}

@article{kahneman2006anomalies,
  title={Anomalies: Utility maximization and experienced utility},
  author={Kahneman, Daniel and Thaler, Richard H},
  journal={Journal of Economic Perspectives},
  volume={20},
  number={1},
  pages={221--234},
  year={2006},
  publisher={American Economic Association}
}

@inproceedings{anunrojwong2022robustness,
  title={On the robustness of second-price auctions in prior-independent mechanism design},
  author={Anunrojwong, Jerry and Balseiro, Santiago and Besbes, Omar},
  booktitle={Proceedings of the 23rd ACM Conference on Economics and Computation},
  pages={151--152},
  year={2022}
}

@article{bhardwaj2025algorithms,
  title={Algorithms for min-buying in networks},
  author={Bhardwaj, Aaditya and Black, Ben and Dokka, Trivikram and Kirkbride, Christopher},
  journal={arXiv preprint arXiv:2502.14459},
  year={2025}
}

@article{aouad2018approximability,
  title={The approximability of assortment optimization under ranking preferences},
  author={Aouad, Ali and Farias, Vivek and Levi, Retsef and Segev, Danny},
  journal={Operations Research},
  volume={66},
  number={6},
  pages={1661--1669},
  year={2018},
  publisher={INFORMS}
}

@article{simon1955behavioral,
  title={A behavioral model of rational choice},
  author={Simon, Herbert A},
  journal={The Quarterly Journal of Economics},
  pages={99--118},
  year={1955},
  publisher={JSTOR}
}

@article{gigerenzer1996reasoning,
  title={Reasoning the fast and frugal way: models of bounded rationality.},
  author={Gigerenzer, Gerd and Goldstein, Daniel G},
  journal={Psychological Review},
  volume={103},
  number={4},
  pages={650--669},
  year={1996},
  publisher={American Psychological Association}
}

@article{bowers2025polynomial,
  title={Polynomial-Time Approximation Schemes via Utility Alignment: Unit-Demand Pricing and More},
  author={Bowers, Robin and Garbea, Marius and Pountourakis, Emmanouil and Taggart, Samuel},
  journal={arXiv preprint arXiv:2506.20030},
  year={2025}
}

@article{milgrom1982theory,
  title={A theory of auctions and competitive bidding},
  author={Milgrom, Paul R and Weber, Robert J},
  journal={Econometrica: Journal of the Econometric Society},
  pages={1089--1122},
  year={1982},
  publisher={JSTOR}
}

@article{karlin1981total,
  title={Total positivity properties of absolute value multinormal variables with applications to confidence interval estimates and related probabilistic inequalities},
  author={Karlin, Samuel and Rinott, Yosef},
  journal={The Annals of Statistics},
  pages={1035--1049},
  year={1981},
  publisher={JSTOR}
}

@article{kahneman2003maps,
  title={Maps of bounded rationality: Psychology for behavioral economics},
  author={Kahneman, Daniel},
  journal={American Economic Review},
  volume={93},
  number={5},
  pages={1449--1475},
  year={2003},
  publisher={American Economic Association}
}

@article{karle2015loss,
  title={Loss aversion and consumption choice: Theory and experimental evidence},
  author={Karle, Heiko and Kirchsteiger, Georg and Peitz, Martin},
  journal={American Economic Journal: Microeconomics},
  volume={7},
  number={2},
  pages={101--120},
  year={2015},
  publisher={American Economic Association 2014 Broadway, Suite 305, Nashville, TN 37203-2425}
}

@article{stuttgen2012satisficing,
  title={A satisficing choice model},
  author={St{\"u}ttgen, Peter and Boatwright, Peter and Monroe, Robert T},
  journal={Marketing Science},
  volume={31},
  number={6},
  pages={878--899},
  year={2012},
  publisher={INFORMS}
}

@article{goyal2016near,
  title={Near-optimal algorithms for the assortment planning problem under dynamic substitution and stochastic demand},
  author={Goyal, Vineet and Levi, Retsef and Segev, Danny},
  journal={Operations Research},
  volume={64},
  number={1},
  pages={219--235},
  year={2016},
  publisher={INFORMS}
}

@article{cao2022network,
  title={Network revenue management under a spiked multinomial logit choice model},
  author={Cao, Yufeng and Kleywegt, Anton J and Wang, He},
  journal={Operations Research},
  volume={70},
  number={4},
  pages={2237--2253},
  year={2022},
  publisher={INFORMS}
}

@article{chen2025approximation,
  title={Approximation Schemes for Dynamic Pricing with Opaque Products},
  author={Chen, Xingxing and DeValve, Levi and Feldman, Jacob and Huang, Yukai and Segev, Danny},
  journal={Available at SSRN 5400432},
  year={2025}
}

@article{housni2025price,
  title={Price and assortment optimization under the multinomial logit model with opaque products},
  author={El Housni, Omar and Elmachtoub, Adam N and Sheth, Harsh and Shi, Jiaqi},
  journal={arXiv preprint arXiv:2502.08124},
  year={2025}
}

@article{milgrom1981good,
  title={Good news and bad news: Representation theorems and applications},
  author={Milgrom, Paul R},
  journal={The Bell Journal of Economics},
  pages={380--391},
  year={1981},
  publisher={JSTOR}
}

@article{karlin1980classes,
  title={Classes of orderings of measures and related correlation inequalities. {I}. Multivariate totally positive distributions},
  author={Karlin, Samuel and Rinott, Yosef},
  journal={Journal of Multivariate Analysis},
  volume={10},
  number={4},
  pages={467--498},
  year={1980},
  publisher={Elsevier}
}

%-------------------------------------------------------
\newpage
\ECSwitch
\renewcommand{\theHsection}{EC.\arabic{section}}
\renewcommand{\theHequation}{EC.\arabic{equation}}
\renewcommand{\theHlemma}{EC.\arabic{lemma}}
\renewcommand{\theHexample}{EC.\arabic{example}}

\begin{APPENDICES}

\section{Omitted Examples}\label{sec:examples}

In this section, we include all examples which were omitted from the main body.

\example[Exchangeable \& Affiliated Valuations, Quasi-Linear Case]
\label{ex:affiliated_ql_gap}
In this example, we give an explicit construction and analysis of an instance where valuations are drawn from an exchangeable and affiliated distribution and where, for quasi-linear utility-maximizing customers, discriminatory pricing earns more than twice that of the optimal uniform pricing. Thus, the separation between the uniform pricing and discriminatory pricing in the quasi-linear case is strictly greater for exchangeable and affiliated distributions than it is in the IID setting. 

Suppose there are $n = 56$ items partitioned into three groups of items, $G_1$, $G_2$, and $G_3$, where $|G_1| = 1$, $|G_2| = 5$, and $|G_3| = 50$. Let customers be quasi-linear utility maximizers who form valuations by first drawing a latent type $T \in \{1, 2, 3\}$ according to $\Pr(T = 1) = \frac{99}{100}$, $\Pr(T = 2) = \frac{99}{10000}$, and $\Pr(T = 3) = \frac{1}{10000}$. Conditional on $T$, the customers draw valuations IID according to one of three distributions. For a type-1 customer, valuations are always equal to one, i.e., $V_i = 1$ with probability 1. For a type-2 customer, valuations are drawn according to a binary distribution such that $V_i = 1$ with probability $\frac{7}{10}$ and otherwise $V_i = 100$ with probability $\frac{3}{10}$. Finally, for a type-3 customer, valuations are drawn IID according to:
\[
    V_i=
    \begin{cases}
        10000 & \text{with probability } \frac1{25},\\[1mm]
        100    & \text{with probability } \frac{36}{125},\\[1mm]
        1       & \text{with probability } \frac{84}{125}.
    \end{cases}
\]
Thus overall, customers draw valuations according to a mixture of IID distributions, and the mixture family satisfies a monotone likelihood ratio property and hence is exchangeable and affiliated (this is easy to verify directly, we omit the computation for brevity). 

Now, over this instance we will first bound the revenue of uniform pricing. If $p \le 1$, then $\rev_{UP}(p) = p \le 1$ since valuations are always at least one, and this is tight when $p = 1$. If $p \in (1, 100]$, then only type-2 and type-3 customers can purchase and since these types have total probability $1/100$, $\rev_{UP}(p) \le \frac{p}{100} \le 1$. If $p \in (100, 10000]$, then only type-3 customers can purchase and since these occur with probability 1/10000, $\rev_{UP}(p) \le \frac{p}{10000} \le 1$. Finally, prices over 10000 earn no revenue, thus we conclude that $\rev_{UP} = 1$.

Now consider the feasible discriminatory pricing that charges $p_1 = .99$ for items in group $G_1$, $p_2 = 99.89$ for items in group $G_2$, and $p_3 = 9900.89$ for items in group $G_3$. A type-1 customer always purchases the item in $G_1$, contributing expected revenue of $\frac{99}{100}(0.99) = 0.9801$. A type-2 customer purchases an item in $G_2$ whenever at least one of the items in $G_2$ is valued at $100$ and all items in $G_1$ are valued at 1, otherwise they purchase an item in $G_1$. The expected revenue from a type-2 customer is:
\
\begin{equation*}
    \frac{99}{10000}\left( (99.89) \frac{7}{10}\left(1-\left(\frac{7}{10}\right)^5\right) 
    + .99\left(1 - \frac{7}{10}\left(1-\left(\frac{7}{10}\right)^5\right)\right) \right) \ge .5799.
\end{equation*}

Finally, a type-3 customer purchases an item in $G_3$ whenever none of
the six items in $G_1 \cup G_2$ are valued at $10000$ and at least one
item in $G_3$ is valued at $10000$. A lower-bound on the expected revenue from type-3 customers, considering only this purchasing event and otherwise assuming they buy an item in $G_1$ is:
\
\begin{equation*}
    \frac{1}{10000}\left(9900.89 \left(\frac{24}{25}\right)^6\left(1-\left(\frac{24}{25}\right)^{50} \right) + .99\left(1- \left(\frac{24}{25}\right)^6\left(1-\left(\frac{24}{25}\right)^{50} \right)\right)\right) \ge .6743
\end{equation*}
Combining the revenue from the three types, we have that $\rev_{DP} \ge .9801 + .5799 + .6743 = 2.2343$. Thus discriminatory pricing can earn strictly more than twice as much revenue as uniform pricing when valuations are exchangeable and affiliated, separating it from the IID case for quasi-linear customers. 
% Furthermore, for simplicity we have limited this construction to three customer types but extending the construction allows us to achieve any level of gap between uniform pricing and discriminatory pricing revenue. 
\hfill\Halmos\endexample

%%%%%%%%%%%%%%%%%%%%%%%%%%%%%%%%%%%%%%%%%%%%%
\vskip 10pt
%%%%%%%%%%%%%%%%%%%%%%%%%%%%%%%%%%%%%%%%%%%%%

\example[General Independent Valuations, Min-Buying Case] \label{ex:notReg}
In this example, we give an explicit construction and analysis of an instance with independent valuations and min-buying customers where a single price can guarantee no better than $\frac{1}{n}$ of the revenue of the optimal differentiated pricing. We note that a similar construction appears in Section 5 of \cite{hartline2009simple}, albeit in a slightly different context. 

Fix any $\ep \in (0, 1)$ and choose some $B > 2/\ep$. Suppose customers are min-buying and draw valuations $V_i$ for items $i \in [n]$ independently according to two-point distributions where $V_i = 0$ with probability $1 - B^{-i}$ and $V_i = B^i$ with probability $B^{-i}$. For this instance, the optimal uniform price must be in the support of the valuations. To bound the optimal uniform pricing revenue, consider the price $p = B^i$ for some $i \in [n]$. This uniform price earns:
\
\begin{align*}
    \rev_{UP}(B^i)
    = B^i \Pr\left(\max_j V_j \ge B^i\right)
    = B^i \left(1 - \prod_{j = i}^n (1 - B^{-j}) \right)
    \le B^i \sum_{j = i}^n B^{-j}
    = \sum_{j = 0}^{n-i} B^{-j}
    \le \frac{B}{B-1},
\end{align*}
where the first inequality follows from the union bound and the second by adding terms and summing the geometric series. Thus, for any price $p$ in the support, $\rev_{UP}(p) \le \frac{B}{B-1}$.

Now consider the differentiated pricing which uses prices $\p := (B, B^2, \ldots, B^n)$. Under this pricing, a min-buying customer buys item $i$ with probability $B^{-i} \prod_{j = 1}^{i - 1}\left(1- B^{-j}\right)$, and the revenue of differentiated pricing under $\p$ can be bounded as
\
\begin{align*}
    \rev_{DP}(\p) 
    &= \sum_{i=1}^n B^i B^{-i} \prod_{j = 1}^{i-1}(1- B^{-j}) \\
    &= \sum_{i=1}^n \prod_{j=1}^{i-1}(1-B^{-j}) \\
    &\ge \sum_{i=1}^n \prod_{j=1}^{\infty}(1-B^{-j})\\
    &\ge \sum_{i=1}^n (1 - \sum_{j=1}^{\infty}B^{-j})\\
    &= \sum_{i=1}^n (1 - \frac{1}{B-1})
    = n\frac{B-2}{B-1}
\end{align*}
where the second inequality follows from the fact that $\prod_j(1 - x_j) \ge 1 - \sum_j x_j$ for any collection $\{x_j\} \subset [0,1]$. Combining the upper-bound for uniform pricing with the lower bound for differentiated pricing we see $\frac{\rev_{DP}}{\rev_{UP}} \ge n \frac{B-2}{B} = n\left(1 - \frac{2}{B}\right) > n(1-\ep)$ as desired. Further, taking $\ep \to 0$, we obtain that the worst-case approximation factor of uniform pricing for general independent valuations is exactly $n$.
\hfill\Halmos\endexample

\section{Omitted Proofs}\label{sec:proofs}

In this section, we include proofs of all results which were omitted from the main body.

\subsection{Omitted Proofs from \cref{sec:symm}}
\begin{proof}{Proof of \cref{thm:exchang_aff}.}
We will prove the stronger statement that, for every exchangeable and affiliated joint distribution $F$ over valuations $V_1, \ldots, V_n$, and for every price vector $\bp = (p_1, \ldots, p_n)$, there exists a price $\hat p \in \{p_1, \ldots, p_n\}$ such that $\rev_{UP}(\hat p) \ge \rev_{DP}(\bp)$. As in the proof of \cref{thm:iid}, our proof will proceed by induction on the number of items, $n$. The base case $n = 1$ follows immediately, so suppose $n \ge 2$ and that the claim holds for $n - 1$ items under every exchangeable and affiliated distribution. Fix some arbitrary price vector $\bp = (p_1, \ldots, p_n)$ and note since the valuation distribution is exchangeable, we may relabel the items so that $p_1 \le p_2 \le \ldots \le p_n$ without loss of generality.

Now, even though valuations are no longer independent, we can still give a recursive analog of \cref{eq:rec_rev} as:
\
\begin{align}\label{eq:single_opt_choice_aff}
    \rev_{DP}(\bp; F) &=
    p_1 \Pr(V_1\ge p_1) + \Pr(V_1 < p_1) \rev_{DP}^{\overline S_1}
    \left((p_2, \ldots, p_n); F^{V_1 < p_1} \right),
\end{align}
where we will use the notation $\rev_{DP}(\cdot \, ; \, G)$ to emphasize that the revenue is taken as an expectation over valuations $\bV \sim G$, and where $F^{V_1 < p_1}$ denotes the conditional distribution of $(V_2, \ldots, V_n)$ after observing the event $V_1 < p_1$. We may also suppose that $\Pr(V_1 < p_1) > 0$ without loss of generality since, if $\Pr(V_1 < p_1) = 0$, then the lowest priced first item is always purchased and $\rev_{DP}(\bp) = \rev_{UP}(p_1)$. 

Now, in order to apply our induction hypothesis in \cref{eq:single_opt_choice_aff}, we must show that $F^{V_1<p_1}$ is also exchangeable and affiliated. To see this, recall that affiliation is preserved under products, conditioning on sublattice events, and marginalization \citep{karlin1980classes}, and note that, up to normalization, the density of $F^{V_1<p_1}$ is:
\[
    f^{V_1 < p_1}(v_2, \ldots, v_n) \propto \int_0^{p_1} f(x, v_2, \ldots, v_n) dx.
\]
The set $\left\{(v_1, \ldots, v_n) | v_1 < p_1 \right\}$ is a sublattice, hence its indicator satisfies the affiliation inequality in \cref{eq:ea-refined-affiliation}, and thus $f(v_1, \ldots, v_n) \I\{v_1 < p_1\}$ is affiliated. Marginalizing this function over $v_1$ shows that the unnormalized density $f^{V_1<p_1}$ is affiliated. Normalization preserves this property, and hence the conditional distribution is affiliated. Furthermore, as the original density $f$ is exchangeable, the conditional density remains exchangeable in the remaining coordinates $2, \ldots, n$.

We may therefore apply the induction hypothesis to the conditional distribution $F^{V_1 < p_1}$ and the truncated price vector $(p_2, \ldots, p_n)$. It follows then that there exists a price $\hat p \in \{p_2, \ldots, p_n\}$ such that:
\
\begin{equation*}
    \rev_{DP}^{\overline S_1}\left((p_2, \ldots, p_n); F^{V_1 < p_1} \right)
    \le \rev_{UP}^{\overline S_1}\left(\hat p ; F^{V_1 < p_1} \right).
\end{equation*}
Note also that since $\hat p \in \{p_2, \ldots, p_n\}$, it implies that $\hat p \ge p_1$. Taking the induction hypothesis and \cref{eq:single_opt_choice_aff} together, we have that:
\begin{align}\label{eq:reduce_two_tiers_aff}
    \rev_{DP}(\bp;F)
    \le p_1\Pr(V_1\ge p_1) + \Pr(V_1 < p_1) \rev_{UP}^{\overline S_1} \left( \hat p; F^{V_1 < p_1}\right) 
    = \rev_{DP}(\tilde\bp ; F),
\end{align}
where $\tilde\bp := (p_1, \hat p\e_{n-1})$. Thus, the original price vector $\bp$ is weakly dominated by the two-tier price vector $\tilde\bp$. To complete the proof, it remains to show that:
\begin{equation} \label{eq:two_tier_goal_aff}
    \rev_{DP}(\tilde\bp;F) \le \max\left\{\rev_{UP}(p_1;F), \rev_{UP}(\hat p;F) \right\}.
\end{equation}

To this end, define the following four events:
\begin{align*}
    &A_1 := \left\{\max_{i \in [n]} V_i \ge p_1 \right\},
    &&A_2 := \left\{\max_{i \in [n]} V_i \ge \hat p \right\}, \\
    &B_1 := \left\{ V_1 < p_1, \max_{i \in \{2, \ldots, n\}} V_i \ge p_1 \right\},
    && B_2 := \left\{V_1 < p_1, \max_{i \in \{2, \ldots, n\}} V_i \ge \hat p \right\}.
\end{align*}
$A_1$ and $A_2$ are the events in which an item is purchased under uniform prices $p_1$ and $\hat p$, respectively. $B_1$ and $B_2$ are the events that occur when the first item is rejected at price $p_1$ and another item is acceptable at price $p_1$ or $\hat p$, respectively. Note we may suppose $\Pr(A_2) > 0$ since, if not, the high-price tier never sells which implies that:
\begin{equation*}
    \rev_{DP}(\tilde\bp;F) = p_1 \Pr(V_1 \ge p_1) \le p_1\Pr(A_1) = \rev_{UP}(p_1;F),
\end{equation*} 
completing the proof. Also note that $\Pr(A_2) > 0$ implies that $\Pr(A_1) > 0$ because $A_2 \subseteq A_1$.

Now, following the proof of \cref{thm:iid}, define $\lambda := \frac{\Pr(V_1 \ge p_1)}{\Pr(A_1)}$. Since $A_1$ is the disjoint union of $\{V_1 \ge p_1\}$ and $B_1$, we have that $1 - \lambda = \frac{\Pr(B_1)}{\Pr(A_1)}$. Consequently,
\begin{align}\label{eq:goal_aff}
    \max\left\{\rev_{UP}(p_1;F), \rev_{UP}(\hat p;F)\right\} 
    &\ge \lambda \rev_{UP}(p_1; F) + (1 - \lambda)\rev_{UP}(\hat p; F) \notag\\
    &\quad= p_1 \Pr(V_1 \ge p_1) + \hat p \frac{\Pr(B_1)}{\Pr(A_1)} \Pr(A_2).
\end{align}
To complete the proof, it suffices to show that $\frac{\Pr(B_1)}{\Pr(A_1)} \Pr(A_2) \ge \Pr(B_2)$ since $\rev_{DP}(\tilde\bp;F) = p_1 \Pr(V_1 \ge p_1) + \hat p \Pr(B_2)$. By rearranging and conditioning, this inequality is equivalent to showing that:
\
\begin{equation}\label{eq:star_aff}
    \Pr(V_1 < p_1 \mid A_1) \ge \Pr(V_1 < p_1 \mid A_2).
\end{equation}
The remainder of the proof will be to show \cref{eq:star_aff}. Before proceeding, let us pause to note that \cref{eq:star_aff} is intuitive for exchangeable and affiliated valuation distributions. It states that the higher the lower bound imposed on the customer's maximum valuation is (here, $\hat p$ rather than $p_1$), the more the conditional marginal distribution of each item valuation should shift upward. In particular, $V_1$ should become less likely to lie below the low-price threshold $p_1$. To formalize this intuition, we will combine exchangeability with the positive-dependence structure implied by affiliation by passing to the joint distribution of the order statistics.

Let $V^{(i)}$ denote the $i^{th}$ order statistic of the valuations, and let $M := V^{(n)}$ denote the maximum. Since the original density $f$ is exchangeable, the joint density of the order statistics is
\begin{equation} \label{eq:order_stat_density_aff}
    g(y_1, \ldots, y_n) := n! f(y_1, \ldots, y_n) \I\{y_1 \le \ldots \le y_n\}.
\end{equation}
The ordered cone $\mathcal C_n := \left\{(y_1, \ldots, y_n) |  y_1 \le \ldots \le y_n \right\}$ is a sublattice. Therefore its indicator function satisfies \cref{eq:ea-refined-affiliation}, and $g(y_1, \ldots, y_n)$ is an affiliated density. Further, again since affiliation is preserved under marginalization, for every $k \in [n - 1]$, the pair $(V^{(k)}, M)$ is affiliated.

Our strategy is now to establish an analogue of \cref{eq:star_aff} for each pair $(V^{(k)}, M)$ and then use an averaging argument to obtain
the desired inequality for $(V_1, M)$. To that end, for every fixed $a$ and every $k \in [n]$, define $\phi_{k,a}(t) := \Pr(V^{(k)} < a \mid M \ge t)$ for all $t$ such that $\Pr(M \ge t) > 0$. We will show that $\phi_{k,a}(t)$ is non-increasing in $t$ for $k \in [n-1]$ (for $k = n$ and $t \ge a$, it is zero). Fix $k \in [n-1]$ and $s < t$ with $\Pr(M \ge t) > 0$, and define the events:
\begin{align*}
    &L := \{V^{(k)} < a\},
    &&H :=\{V^{(k)} \ge a\}, \\
    &C_1 := \{s \le M < t\},
    &&C_2 := \{M \ge t\}.
\end{align*}
If $\Pr(C_1) = 0$, then $\phi_{k,a}(s) = \phi_{k,a}(t)$, so the desired inequality holds. Henceforth suppose $\Pr(C_1) > 0$. Let $h_k(x,z)$ denote the joint density of $(V^{(k)}, M)$. Since $h_k$ is affiliated, for every $x < a \le y$ and $s \le z < t \le w$, we have
\[
    h_k(x,z)h_k(y,w) \ge h_k(x,w)h_k(y,z).
\]
Integrating this inequality over $x < a$, $y \ge a$, $s \le z < t$, and $w \ge t$ gives:
\begin{align*}
    &\int_{x<a}\int_{y\ge a}\int_s^t\int_t^\infty
    h_k(x,z)h_k(y,w)
    \,dw\,dz\,dy\,dx\\
    &\qquad\ge
    \int_{x<a}\int_{y\ge a}\int_s^t\int_t^\infty
    h_k(x,w)h_k(y,z)
    \,dw\,dz\,dy\,dx.
\end{align*}
The left-hand integral factors over $(x,z)$ and $(y,w)$, while the right-hand integral factors over $(x,w)$ and $(y,z)$. Therefore, this inequality is equivalent to:
\
\begin{equation}\label{eq:tp2_rectangle_aff}
    \Pr(L\cap C_1)\Pr(H\cap C_2) \ge \Pr(L\cap C_2)\Pr(H\cap C_1),
\end{equation}
which is equivalent to $\Pr(L \mid C_1) \ge \Pr(L \mid C_2)$ by conditioning. Moreover, note $\{M \ge s\}$ is the disjoint union of $C_1$ and $C_2$. Thus $\Pr(L \mid M \ge s)$ is a convex combination of $\Pr(L\mid C_1)$ and $\Pr(L \mid C_2)$. It follows that $\Pr(L \mid M \ge s) \ge \Pr(L \mid C_2) = \Pr(L \mid M \ge t)$. Thus, $\phi_{k,a}(s) \ge \phi_{k,a}(t)$, and we conclude that $\phi_{k,a}$ is non-increasing.

Now we will translate these order-statistic comparisons back to an individual valuation. By exchangeability, conditional on the symmetric event $\{M \ge t\}$, every coordinate has the same probability of lying below $a$. Therefore,
\begin{align}\label{eq:coordinate_order_stat_aff}
    \Pr(V_1 < a \mid M\ge t) 
    = \frac{1}{n} \E\left[\sum_{i=1}^n \mathbf 1\{V_i<a\} \middle| M \ge t \right]
    = \frac{1}{n} \sum_{k=1}^n \phi_{k, a}(t)
    = \frac{1}{n} \sum_{k=1}^{n-1} \phi_{k, a}(t),
\end{align}
where the final equality follows from noting that when $t \ge a$, the event $\{M \ge t\}$ implies that $M \ge a$, and the final term is zero. As every term on the right-hand side is non-increasing in $t$, it follows that the probability $\Pr(V_1 < a \mid M \ge t)$ is also non-increasing for $t\ge a$. Taking $a = p_1$ and $t \in \{p_1, \hat p\}$ gives $\Pr(V_1 < p_1 \mid A_1) \ge \Pr(V_1 < p_1 \mid A_2)$ which proves \cref{eq:star_aff}. Substituting this inequality into \cref{eq:goal_aff} and combining it with \cref{eq:reduce_two_tiers_aff}, gives that $\rev_{DP}(\bp ; F) \le \max\left\{\rev_{UP}(p_1;F), \rev_{UP}(\hat p ; F) \right\}$ which completes the induction. 

Thus every discriminatory price vector is weakly dominated by a uniform price chosen from its own price vector, and this implies that $\rev_{DP} =\rev_{UP}$ as desired.
\hfill\Halmos\end{proof}

%%%%%%%%%%%%%%%%%%%%%%%%%%%%%%%%%%%%%%%%%%%%%
\vskip 10pt
%%%%%%%%%%%%%%%%%%%%%%%%%%%%%%%%%%%%%%%%%%%%%

\begin{proof}{Proof of \cref{thm:exchang}.}
As the upper-bound is handled by Theorem 3(i) of \cite{elmachtoub2021power}, here we will construct a lower-bounding instance where the gap in revenue between discriminatory pricing and uniform pricing for min-buying customers is logarithmic. To that end, suppose there are $n$ items and for each $i \in [n]$, define $r_i := 4^{n - i}$, $\lambda_i := \frac{3}{8 r_i}$, and $k_i := \max\{1, \left\lfloor\frac{n}{2 (n - i + 1)}\right\rfloor\}$. Consider the distribution of valuations $\bV = (V_1, \ldots, V_n) \sim F$ such that there are $n$ valuation types. Each type $i \in [n]$ occurs with probability $\lambda_i$ and a type-$i$ valuation vector is constructed by choosing a subset $S \subseteq [n]$ of cardinality $k_i$ uniformly at random among all subsets, then setting $V_j = 2 r_i$ if $j \in S$ and $V_j = 0$ if $j \notin S$. With the remaining probability $1 - \sum_i \lambda_i$, let $\bV = \textbf{0}$. Note this distribution over valuations is clearly exchangeable and that this is a valid probability distribution since,
\
\begin{equation*}
    \sum_{i=1}^n \lambda_i = \frac{3}{8}\sum_{i=1}^n 4^{i - n} 
    = \frac{3}{8} \sum_{i = 0}^{n - 1} 4^{-i} 
    = \frac{1}{2}\left(1 - 4^{-n}\right) < 1.
\end{equation*}

% For convenience, define \(r_{N+1}=0\).

Now to upper-bound the revenue of an optimal uniform pricing, consider any uniform price $p \in (0, 2 r_1]$  (since if $p > 2 r_1$ no buyer purchases). Let $i \in [n]$ be the unique index such that $2 r_{i+1} < p \le 2 r_i$ (letting $r_{n+1} = 0$). For any type-$t$ valuation vector where $t > i$, we have
\[
    \max_{j \in [n]} V_j = 2 r_t \le 2r_{i+1} < p
\]
and thus they will not purchase any item at price $p$. Consequently,
\
\begin{align*}
    \rev_{UP}(p) & \le p \sum_{j = 1}^i \lambda_j 
    \le 2 r_i \sum_{j = 1}^i \frac{3}{8 r_j} 
    = \frac{3}{4} \sum_{j = 1}^i\frac{r_i}{r_j} 
    = \frac{3}{4} \sum_{j = 0}^{i - 1} 4^{-j}
    = 1 - 4^{-i} < 1,
\end{align*}
which implies that $\max_p \rev_{UP}(p) \le 1$. 

Now we will lower-bound the revenue of optimal discriminatory pricing by considering the policy which uses prices $\bp = (r_1, r_2, \ldots, r_n)$, indexed from high to low. Consider a type-$i$ valuation vector, i.e., the valuation vector where $|S| = k_i$ items have valuation $2 r_i$ chosen u.a.r., and consider the event where the highest index item with non-zero valuation is $i$, denoted by $E_i$. Under event $E_i$, the item $i$ is valued as $2 r_i$ which exceeds its price of $r_i$, and it is the cheapest item which gives non-negative utility thus item $i$ is purchased. Further, note that since $S$ is a uniformly random, size $k_i$ subset of $[n]$,
\
\begin{equation}\label{eq:binom_lb}
    \Pr(E_i \mid \text{Valuation of type $i$}) = \frac{\binom{i-1}{k_i-1}}{\binom{n}{k_i}}.
\end{equation}
Now, if we can prove that $\Pr(E_i \mid \text{Valuation of type $i$}) \ge \frac{1}{4 (n - i + 1)}$, it will imply our desired result since $\rev_{UP} \le 1$ and
\
\begin{align*}
    \rev_{DP} \ge \rev_{DP}(r_1, \ldots, r_n) 
    &\ge \sum_{i=1}^n r_i \Pr(E_i \mid \text{Valuation of type $i$}) \lambda_i \\
    &\ge \sum_{i=1}^n r_i \frac{1}{4(n - i + 1)} \frac{3}{8 r_i} \\
    &= \frac{3}{32} \sum_{i=1}^n \frac{1}{n-i+1} \\
    &= \frac{3}{32} H_n,
\end{align*}
where $H_n$ is the $n^{th}$ harmonic number. Thus, in the remainder of this proof we will focus on lower-bounding the expression in \cref{eq:binom_lb}. 

We will consider two cases. First suppose that $i$ is such that $k_i = 1$. In this case, $\Pr\left(E_i \mid \text{Valuation of type $i$}\right) =  \frac{1}{n}$ by \cref{eq:binom_lb}, and  $k_i = 1$ implies that $\frac{n}{2(n - i + 1)} < 2$ by definition. Thus $\Pr\left(E_i \mid \text{Valuation of type $i$} \right) = \frac{1}{n}  > \frac{1}{4(n - i + 1)}$ as desired. Now suppose $i$ is such that $k_i \ge 2$. Since $k_i = \left\lfloor\frac{n}{2(n - i + 1)}\right\rfloor$, we have $k_i \le \frac{n}{2(n - i + 1)} < k_i + 1$. By \cref{eq:binom_lb} we obtain:
\
\begin{align*}
    \Pr\left(E_i \mid \text{Valuation of type $i$}\right) 
    &= \frac{k_i}{n} \prod_{j = 1}^{k_i - 1} \frac{i - j}{n - j} \\
    &= \frac{k_i}{n} \prod_{j = 1}^{k_i - 1} \left(1 - \frac{n - i}{n - j}\right) \\
    &\ge \frac{k_i}{n} \left(1 - \sum_{j=1}^{k_i - 1} \frac{n - i}{n - j}\right) \\
    &\ge \frac{k_i}{n} \left(1 - \frac{(k_i - 1)(n - i)}{n - k_i + 1} \right) \\
    &\ge \frac{k_i}{n} \left(1 - \frac{k_i - 1}{2 k_i}\right) \\
    &\ge \frac{k_i}{2 (n - i + 1) (k_i + 1)} \left(1 - \frac{k_i - 1}{2 k_i}\right) \\
    &\ge \frac{1}{4 (n - i + 1)}
\end{align*}
where the first inequality follows from $\prod_j (1 - a_j) \ge 1 - \sum_j a_j$ for any numbers $\{a_j\} \subset [0,1]$, the third inequality follows from $2k_i(n - i) \le n - 2k_i \le n - k_i + 1$, the fourth inequality follows from $\frac{n}{2(n - i + 1)} < k_i + 1$, and the final inequality is by rearranging. Thus $\Pr\left(E_i \mid \text{Valuation of type $i$}\right) \ge \frac{1}{4(n - i + 1)}$ for all $i$, as desired. 
\hfill\Halmos\end{proof}

\subsection{Omitted Proofs from \cref{sec:reg}}
\label{sec:omitted_pfs_reg}

\begin{proof}{Proof of \cref{lem:reg-refined-ratio-reduction}.}
Suppose the valuations $V_i \sim F_i$, $i \in [n]$, are independent, continuous, and regular. Let $\p = (p_1, \ldots, p_n)$ be a revenue-maximizing differentiated pricing, and relabel the items so that $p_1 \le \ldots \le p_n$. Define $q_i := \Pr(V_i \ge p_i)$ as the acceptance probability at price $p_i$. We will construct an instance of triangular-revenue distributions whose optimal differentiated-pricing revenue is at least $\rev_{DP}$, while its optimal uniform-pricing revenue is no greater than that of the original instance.

Note that WLOG we may assume $q_i \in (0,1)$ for all $i \in [n]$. Indeed, if $q_i = 0$, then item $i$ never sells under $\p$ and can be omitted without changing the revenue of this price vector or increasing the uniform-pricing revenue. If $q_i = 1$ for some $i$, then a min-buying customer never pays more than $p_i$. Since $\p$ is optimal, $\rev_{DP} \le p_i = \rev_{UP}(p_i) \le \rev_{UP} \le \rev_{DP}$. Thus, $\rev_{DP} = \rev_{UP} = p_i$. A single item with distribution $\Tri(2p_i, 2p_i)$ has optimal uniform- and differentiated-pricing revenues equal to $p_i$, giving the required replacement. 

Now, for each item $i$, define $a_i := p_i q_i/(1-q_i)$ and replace $V_i$ by an independent valuation distribution $\widetilde V_i \sim \Tri(a_i, p_i)$, with CDF $\widetilde F_i$. By construction, the acceptance probability at the price $p_i$ is unchanged since $\Pr(\widetilde V_i \ge p_i) = \frac{a_i}{a_i+p_i} = q_i$. Next, recall that differentiated-pricing revenue under $\p$ is:
\[
    \rev_{DP}(\p) = \sum_{i=1}^n p_i q_i\prod_{j=1}^{i-1}(1-q_j).
\]
Since the prices and acceptance probabilities are unchanged, the transformation preserves the revenue of the designated price vector. Thus, $\widetilde{\rev}_{DP} \ge \rev_{DP}$, and it only remains to show that uniform-pricing revenue is weakly decreased over this new instance. 

To see this, note that by definition the revenue-quantile function for triangular-revenue distributions, $\widetilde F_i = \Tri(a_i, p_i)$ is such that for any $x \in (q_i, 1)$, $\rev_{\widetilde F_i}(x) = a_i(1-x) = p_i q_i \frac{1-x}{1-q_i}$. Further, because $F_i$ is regular, the concavity of its revenue-quantile function implies:
{\small
\begin{align*}
    \rev_{F_i}(x) = x F_i^{-1}(1-x) \ge
    \frac{1-x}{1-q_i}\rev_{F_i}(q_i) + \frac{x - q_i}{1-q_i}\rev_{F_i}(1)
    \ge \frac{1-x}{1-q_i}\rev_{F_i}(q_i)
    = p_i q_i \frac{1 - x}{1 - q_i} = \rev_{\widetilde F_i}(x) = x \widetilde F_i^{-1}(1 - x),
\end{align*}
}
where the first inequality follows from Jensen's inequality and the second inequality uses $\rev_{F_i}(1) \ge 0$. Dividing both sides of the above by $x > 0$ gives $F_i^{-1}(1-x) \ge \widetilde F_i^{-1}(1-x)$ for all $i \in [n]$. Now fix some uniform price $0 < \hat p < p_i$, and consider $x: = a_i/(a_i + \hat p) \in (q_i, 1)$. Then $\widetilde F_i^{-1}(1 - x) = \hat p$, and $\Pr(V_i \ge \hat p) \ge x = \Pr(\widetilde V_i \ge \hat p)$. Further note, if $\hat p = p_i$, the acceptance probabilities are equal by construction, at $\hat p = 0$ both acceptance probabilities are one, and for $\hat p > p_i$ for $\hat p > p_i$ the transformed item's acceptance probability is zero. Consequently, $\Pr(\widetilde V_i \ge \hat p) \le \Pr(V_i \ge \hat p)$ for every $\hat p \ge 0$.

By independence, the transformation therefore weakly decreases the probability of sale at every uniform price. Hence $\widetilde{\rev}_{UP} \le \rev_{UP}$, completing the proof.
\hfill\Halmos\end{proof}

%%%%%%%%%%%%%%%%%%%%%%%%%%%%%%%%%%%%%%%%%%%%%
\vskip 10pt
%%%%%%%%%%%%%%%%%%%%%%%%%%%%%%%%%%%%%%%%%%%%%

\begin{proof}{Proof of \cref{lem:reg-refined-canonicalization}.}
Fix a triangular-revenue instance $\cI$ such that the revenue of the optimal uniform pricing over that instance is equal to 1 and suppose $\p =  (p_1, \ldots, p_n)$ is an optimal differentiated pricing for $\cI$. For ease of exposition, we will use the notation $\rev_{DP}(\p; \cI)$ to denote the revenue of the price vector $\p$ over the instance $\cI$, and use $\rev_{DP}(\cI)$ and $\rev_{UP}(\cI)$ for the optimal differentiated- and uniform-pricing revenues over $\cI$, respectively. We will show that $\cI$ can always be transformed into a canonical instance over which uniform-pricing earns revenue of at most one and the differentiated pricing which sets the price of each item at its cap (what we will call \textit{cap-pricing}) earns revenue within $\ep > 0$ of the original differentiated-pricing optimum. We note that some steps of this proof parallel the triangular-revenue reductions in \citet{alaei2019optimal,jin2020tight}, however the details here differ (and often simplify) as the differentiated-pricing revenue in our model is governed by the min-buying continuation recursion.

First, we reduce to the instance such that each item is priced at its cap, i.e., $p_i = \v_i$. Consider an item $V_i \sim \Tri(a_i, \v_i)$ offered at price $p_i$. An item priced either at zero or above its cap never earns revenue or sells, respectively, and may be removed so that we may assume that $0 < p_i \le \v_i$ WLOG. Now consider, for each $i \in [n]$, replacing $V_i$ by $\widehat V_i \sim \Tri(a_i, p_i)$ and note that at the price $p_i$, the acceptance probabilities are preserved, i.e., $\Pr(\widehat V_i \ge p_i) = \frac{a_i}{a_i + p_i} = \Pr(V_i \ge p_i)$. Since the replacement preserves the item's acceptance probability, applying this replacement to each item preserves the total revenue under $\p$. This replacement also weakly reduces uniform-price demand, by a nearly identical argument as in the proof of \cref{lem:reg-refined-ratio-reduction} (we will omit it for brevity). Thus, there is a triangular-revenue instance with the same optimal differentiated pricing revenue as in the original instance where each item is priced at its cap.

Next, we will show that we can aggregate items with a common cap. Having reduced to the instance where items are priced at their caps, now suppose $m$ items have cap $\hat p$ and parameters $a_1, \ldots, a_m$. Consider the transformation that replaces them all by one item with distribution $\Tri(\hat a, \hat p)$, where $\hat a$ is chosen such that $\frac{\hat p}{\hat p + \hat a} = \prod_{j=1}^m \frac{\hat p}{\hat p + a_j}$, or equivalently, 

\begin{equation}\label{eq:rep_a}
    \hat a = \hat p \left(\prod_{j=1}^m \left(1 + \frac{a_j}{\hat p}\right) - 1\right).
\end{equation}
Note then, if $z$ is the continuation revenue from all items more expensive than $\hat p$, the original $m$ items with cap $\hat p$ produce continuation revenue:
\[
    \hat p \left(1 - \prod_{j=1}^m \frac{\hat p}{\hat p + a_j}\right)
    + z \prod_{j=1}^m \frac{\hat p}{\hat p + a_j},
\]
which, by construction, is unchanged by the replacement. It remains to show that the replacement by $\Tri(\hat a, \hat p)$ weakly reduces demand at every lower uniform price. To see this, consider any uniform price $t \in (0, \hat{p}) $, and note the original $m$ items have a combined rejection probability at this price of $\prod_{j=1}^m \frac{t}{t + a_j}$, whereas the new replacement item has rejection probability $\frac{t}{t + \hat a}$. Thus we need only show that $\frac{t}{t + \hat a} \ge \prod_{j=1}^m \frac{t}{t + a_j}$. To this end, let $\lambda := \hat p/t \ge 1$, $b_j := a_j/\hat p$, and plugging these in along with the definition of $\hat a$ in \cref{eq:rep_a}, this comparison can be rewritten as:
\[
    1 + \lambda\left(\prod_{j=1}^m(1 + b_j) - 1\right) \le \prod_{j=1}^m (1 + \lambda b_j).
\]
To verify this inequality, expand both products. Every term of degree $d \ge 1$ is multiplied by $\lambda$ on the left and by $\lambda^d$ on the right. Since $\lambda \ge 1$, the inequality follows. Finally, at $t = 0$, both rejection probabilities are zero and at $t > \hat p$, all $m$ original items and the replacement item are always rejected. Thus, aggregation preserves the cap-pricing revenue and weakly reduces uniform-pricing revenue.

Let $N \le n$ denote the number of remaining items after aggregating all items with common caps, and index their distinct caps in decreasing order, $p_1 > \ldots > p_N > 0$. For the cap-priced and aggregated instance, the continuation revenue can be written recursively as:
\[
    C_0:=0,
    \qquad
    C_k := \frac{a_k + C_{k-1}}{1 + a_k/p_k}, \qquad k \in [N],
\]
where $C_k$ is the revenue generated by the first $k$ items, with the recursion proceeding from the highest cap (and thus price) to the lowest. In particular, $C_N$ is the cap-pricing revenue of the aggregated instance. In the remainder of this proof, we will remove all items whose caps are at most one and then replace the most expensive items with a single high-cap item and at most one boundary item. The resulting instance will retain the desired differentiated-pricing revenue, up to an $\ep$ loss, while its optimal uniform-pricing revenue remains at most one.

To that end, first note we may suppose that $C_N > 1$ since, if $C_N \le 1$, then choosing $P > 1$ sufficiently large that $\frac{P}{P+1} \ge C_N - \ep$ and considering the instance consisting only of one item with valuation $\Tri(1,P)$ has uniform-pricing revenue below one and the desired cap-pricing revenue. Thus it is a canonical instance with $N = 0$ and the result follows in this case. 

Now assume that $C_N > 1$. We begin by removing all items whose caps are at most one. If any such items exist, let $k$ be the first index such that $p_k \le 1$. Note the continuation revenue at $k$ satisfies:
\[
    C_k = \frac{a_k}{a_k + p_k}p_k + \frac{p_k}{a_k + p_k} C_{k-1} \le \max\{p_k, C_{k-1}\},
\]
which, by iterating, implies that for every $j \ge k$, we have $C_j \le \max\{p_k, C_{k-1}\}$. Since $C_N > 1$ and $p_k\le1$, it follows that $C_N \le C_{k-1}$. Thus, removing all items with caps at most one weakly increases the cap-pricing revenue. Furthermore, removing these items does not affect uniform-pricing revenue at prices above one and can only decrease it at prices at most one. Thus the uniform-pricing bound is preserved.

Finally, we now turn to replacing the most expensive items. Let $r$ be the first index at which continuation revenue exceeds one, i.e., such that $C_{r-1} \le 1 < C_r$. Note that $C_r$ is a convex combination of $C_{r-1}$ and $p_r$, and can be written as $C_r = \frac{a_r}{a_r + p_r} p_r + \frac{p_r}{a_r + p_r} C_{r-1}$. Because $C_{r-1} \le 1 < C_r$, it follows that $1 < C_r < p_r$. We will replace the first $r$ items, which have the highest caps and prices, with a single high-cap item and at most one boundary item, preserving their continuation revenue up to an arbitrarily small loss.

Specifically, consider the transformation which replaces the first $r$ items by one item with distribution $\Tri(1, P)$, where $P$ is large enough that $P > p_r$ and $\frac{P}{P + 1} \ge 1 - \ep$, and another item with distribution $\Tri(\bar a, p_r)$, where $\bar a := \frac{p_r (C_r - 1)}{p_r - C_r}$. The continuation revenue of the two items with distribution $\Tri(1, P)$ and $\Tri(\bar a, p_r)$ is:
\[
    \frac{\bar a}{\bar a + p_r} p_r + \frac{p_r}{\bar a + p_r} \frac{P}{P + 1} \ge \frac{\bar a + 1}{1 + \bar a/p_r} - \frac{p_r}{\bar a + p_r} \ep \ge C_r - \ep,
\]
where the inequality follows from recalling $\frac{P}{P + 1} \ge 1 - \ep$. Thus, the replacement preserves the continuation revenue of the first $r$ items up to a loss of $\ep$. Since the revenue of later items only depends on the revenue of the first $r$ items through $C_r$, this replacement also preserves the revenue generated by all items (up to a loss of $\ep$).

Finally we must show this replacement does not improve the uniform pricing revenue. For uniform prices $t > p_r$, the items with caps less than or equal to $p_r$ can be omitted as they will never be purchased. Thus the uniform pricing revenue satisfies $\widetilde \rev_{UP}(t) \le \frac{t}{t + 1} \le 1$ for all $t > p_r$. Now consider uniform prices $t \le p_r$. For uniform pricing revenue to be reduced, the purchasing probability of the items with distributions $\Tri(1, P)$ and $\Tri(\bar a, p_r)$ after the replacement must satisfy:
\
\begin{equation}\label{eq:demand_dominance}
    1 - \frac{t}{t + 1} \frac{t}{t + \bar a} \le 1 - \prod_{i=1}^{r} \frac{t}{t + a_i}.
\end{equation}
To prove \cref{eq:demand_dominance}, define $Q(t) := \prod_{i=1}^{r-1} \frac{t}{t + a_i}$ as the probability that the first $r - 1$ items are all rejected at a uniform price $t$. Since they are triangular-revenue distributions, we can show that $Q(t)$ is bounded by:
\begin{equation}\label{eq:reg-refined-prefix-bound-application}
    Q(t) \le \frac{t}{t + C_{r-1}},
    \qquad 0 < t \le p_r,
\end{equation}
where $C_{r-1}$ is the continuation revenue of the first $r - 1$ items. For brevity, we relegate the proof details of \cref{eq:reg-refined-prefix-bound-application} to \cref{lem:reg-refined-prefix-rejection} in the next subsection. If we can prove $1 - \frac{t}{t + 1} \frac{t}{t + \bar a} \le 1 - \frac{t}{t + C_{r-1}} \frac{t}{t + a_r}$, then we directly obtain \cref{eq:demand_dominance} using \cref{eq:reg-refined-prefix-bound-application}. Therefore, we only need to show $1 - \frac{t}{t + 1} \frac{t}{t + \bar a} \le 1 - \frac{t}{t + C_{r-1}} \frac{t}{t + a_r}$ holds after the replacement, or equivalently,
\
\begin{equation}\label{eq:purchase_prob_dom}
    (t + 1)(t + \bar a) \le (t + C_{r-1})(t + a_r).
\end{equation}

To prove \cref{eq:purchase_prob_dom} is true, we first rewrite $a_r$ and $\frac{p_r}{p_r + a_r}$ as:
\
\begin{equation}\label{eq:formula_ar}
    a_r = \frac{p_r(C_r - C_{r-1})}{p_r - C_r},
    \qquad
    \frac{p_r}{p_r + a_r} = \frac{p_r - C_r}{p_r - C_{r-1}},
\end{equation}
where the equalities follow by reorganizing the continuation revenue equation $C_r = \frac{a_r}{a_r + p_r} p_r + \frac{p_r}{a_r + p_r} C_{r-1}$. Substituting the formulas for $a_r$ and $\bar a$ into $(t + C_{r-1})(t + a_r)$ and $(t + 1)(t + \bar a) $, then gives:
\[
    (t + C_{r-1})(t + a_r) - (t + 1)(t + \bar a) 
    = \frac{1 - C_{r-1}}{p_r - C_r} \left(C_r t + p_r\left(1 + C_{r-1} - C_r\right)\right),
\]
where $1 - C_{r-1} \ge 0$ by the definition of $C_{r - 1}$, and $p_r - C_r > 0$ since $C_r = \frac{a_r}{a_r + p_r} p_r + \frac{p_r}{a_r + p_r} C_{r-1}$ and $p_r > 1$. Now, to show \cref{eq:purchase_prob_dom} is true, we only need to show that $1 + C_{r-1} - C_r \ge 0$.

Note that at uniform price $p_r$, the first $r$ items are the only items that can be purchased as all other items have caps below $p_r$. Since the optimal uniform pricing revenue is at most 1, $\rev_{UP}(p_r) = p_r \left(1 - Q(p_r)\frac{p_r}{p_r + a_r}\right) \le 1$. Together with \cref{eq:reg-refined-prefix-bound-application}, this implies:
\begin{equation}\label{eq:reg-refined-crossing-uniform-bound}
    p_r \left(1 - \frac{p_r}{p_r + C_{r-1}} \frac{p_r}{p_r + a_r}\right) \le 1 \Longleftrightarrow 1 - \frac{1}{p_r} \le \frac{p_r}{p_r + C_{r-1}} \frac{p_r}{p_r + a_r}.
\end{equation}
Substituting the expressions of $a_r$ and $\frac{p_r}{p_r + a_r}$ from \cref{eq:formula_ar} into \cref{eq:reg-refined-crossing-uniform-bound}, we obtain $1 - \frac{1}{p_r} \le \frac{p_r}{p_r + C_{r-1}} \frac{p_r - C_r}{p_r - C_{r-1}}$. Rearranging this inequality yields:
\begin{equation}\label{eq:reg-refined-crossing-bound}
    C_r \le 1 + \frac{(p_r - 1)C_{r-1}^2}{p_r^2} \le 1 + C_{r-1},
\end{equation}
where the last inequality uses the assumption $0 \le C_{r-1} \le 1$ and $p_r > 1$. With \cref{eq:reg-refined-crossing-bound} established, we can now conclude that the replacement preserves the revenue of differentiated pricing at $\p$ (up to a loss of $\ep$) and weakly decreases the optimal uniform pricing revenue.

Finally, we will summarize the transformations above and combine the resulting revenue comparisons to verify that they yield a canonical triangular-revenue distribution. Let $L \ge r$ be the largest index with $p_L > 1$. We replace the original instance $\cI$ with a new instance $\widetilde \cI$ which contains two items with distributions $\Tri(1, P)$ and $\Tri(\bar a, p_r)$ respectively, and the original items $ r+1, \ldots, L$. By the reasoning above, we know that when $P$ is sufficiently large, the revenue of the new instance $\widetilde \cI$ priced at the items' caps $\widetilde \p$ satisfies $\rev_{DP} (\widetilde \p; \widetilde \cI) \ge \rev_{DP}(\cI) - \ep$, and the uniform pricing revenue satisfies $\rev_{UP}(\widetilde \cI) \le \rev_{UP}(\cI)$.
\hfill\Halmos\end{proof}

%%%%%%%%%%%%%%%%%%%%%%%%%%%%%%%%%%%%%%%%%%%%%
\vskip 10pt
%%%%%%%%%%%%%%%%%%%%%%%%%%%%%%%%%%%%%%%%%%%%%

\begin{proof}{Proof of \cref{lem:reg-refined-Tri-bound}.}
Consider a canonical instance where items are valued independently according to $V_0 \sim \Tri(1, P)$, and $V_i \sim \Tri(a_i, \v_i)$, $i \in [N]$, where $P > \v_1 > \ldots > \v_N > 1$, and suppose over this instance that $\rev_{UP} \le 1$. If $N = 0$, its cap-pricing revenue is $P/(P+1) < 1/x^\star$, and the result is immediate. So suppose $N \ge 1$ and define $x_i: = 1/\v_i$ as the reciprocal cap of item $i$, and note $0 < x_1 < \ldots < x_N < 1$. By \cref{eq:continuation}, the continuation revenues for a canonical instance satisfy:
\[
    C_0 = \frac{P}{P + 1} < 1,
    \qquad
    C_k = \frac{a_k + C_{k-1}}{1 + a_k x_k}, \quad k \in [N].
\]
We will show that the cap-pricing revenue, $C_N$, is at most $1/x^\star$. Specifically, we will bound the maximum continuation revenue for any given set of parameters $\{a_1, \ldots, a_k\}$ and $\{x_1, \ldots, x_k\}$ by analyzing the uniform pricing revenue constraints the parameters must satisfy, deriving a relaxed upper-bound, and then analyzing the continuous extension of that upper-bound. 

To begin, consider the constraints induced by the assumption $\rev_{UP} \le 1$. Fix some $k \in [N]$ and consider a uniform price of $\v_k$. At this price, the terminal item (i.e., the item valued according to $V_0 \sim \Tri(1, P)$) is rejected with probability $\v_k/(\v_k + 1)$, and only the terminal item and the first $k$ items can be purchased, since all remaining items have caps below $\v_k$. Thus the uniform pricing revenue at price $\v_k$ is $\rev_{UP}(\v_k) = \v_k \left(1 - \frac{\v_k}{\v_k + 1}\prod_{i=1}^k \frac{\v_k}{\v_k + a_i}\right)$, and it must be less than or equal to 1 which, by rearranging, is equivalent to:
\begin{equation*}
     1 - \frac{1}{\v_k^2} \le \prod_{i=1}^k \frac{\v_k}{\v_k + a_i}.
\end{equation*}
Taking logarithms of both sides of the above, substituting $\v_k$ with $\frac{1}{x_k}$, and considering each cap price gives a system of constraints that the instance parameters must satisfy:
\begin{equation*}
    \sum_{i=1}^k\log(1 + a_i x_k) \le - \log(1 - x_k^2), \quad k \in [N].
\end{equation*}
We can relax and simplify these constraints by noting that since $\log(1 + ax)$ is concave when $a > 0$, by Jensen's inequality for any $x \in [0, 1]$ we have $\log(1 + ax) \ge x \log(1 + a) + (1 - x) \log 1 = x \log(1 + a)$. Applying this bound to the LHS, dividing both sides by $x_k$, and recalling $A(x) := -\frac{\log(1 - x^2)}{x}$ as defined in \cref{sec:independent-regular} yields:
\
\begin{equation}\label{eq:reg-refined-cumulative-budget}
    \sum_{i=1}^k \log(1 + a_i) \le A(x_k), \quad k \in [N].
\end{equation}

Now we will use the system of constraints in \cref{eq:reg-refined-cumulative-budget} to bound $C_N$. For ease of exposition, let $s_0 := 0$ and $s_k := \sum_{i=1}^k \log(1 + a_i)$ denote the LHS partial sums for the $k^{th}$ item in \cref{eq:reg-refined-cumulative-budget} and define the \textit{binding} reciprocal caps by $\widehat x_0 = 0$, and $\widehat x_k := A^{-1}(s_k)$ for $k \in [N]$. By \cref{eq:reg-refined-cumulative-budget} and the fact that $A(x)$ is strictly increasing (this follows directly from the derivative which we will analyze later in the proof), we have $\widehat x_k \le x_k$ for every $k \in [N]$, and $0 = \widehat x_0 < \widehat x_1 < \ldots < \widehat x_N < 1$. Now, define the continuation revenue recursion for the binding reciprocal caps as:
\begin{equation*}
    \widehat C_0 := 1,
    \qquad
    \widehat C_k := \frac{a_k + \widehat C_{k-1}}{1 + a_k \widehat x_k},
    \qquad k \in [N],
\end{equation*}
where $\widehat C_k$ for $k \in [N]$ are the \textit{relaxed} continuation revenues. Since $C_0 < \widehat C_0$ and $\widehat x_k \le x_k$, iterating gives $C_k \le \widehat C_k$ for every $k \in [N]$. Thus, $\widehat C_N$ upper bounds $C_N$. In the remainder of this proof, we will bound $\widehat C_N$ by passing to an appropriate continuous limit.

Specifically, we will construct a continuous upper bound on the relaxed continuation revenues. Note that, when the constraints in \cref{eq:reg-refined-cumulative-budget} all bind, the item parameters are determined by the binding reciprocal caps, i.e., $a_k = \exp\left(A(\widehat x_k) - A(\widehat x_{k-1})\right) - 1$. Thus, choosing an increasing sequence $0 = \widehat x_0 < \widehat x_1 < \ldots < \widehat x_N < 1$ determines both the parameters $a_k$ and the resulting relaxed continuation revenues, which can be expressed through the difference equation:
\[
    \widehat C_k - \widehat C_{k-1} = a_k \bigl(1 - \widehat x_k \widehat C_k \bigr).
\]
We can interpret this in the space of the binding reciprocal caps as associating the interval $[\widehat x_{k-1}, \widehat x_k]$ with item $k$, which contributes a $\log(1 + a_k) = A(\widehat x_k) - A(\widehat x_{k-1})$ increase in the cumulative ``budget'' $A$ over that interval and leads to a marginal change of $a_k \bigl(1 - \widehat x_k \widehat C_k \bigr)$ in relaxed continuation revenue. Using this representation, we can describe how the relaxed continuation revenues evolve as the number of items increases and the corresponding intervals become increasingly small. When consecutive reciprocal caps approach a common value $x$, the ratio $a_k/(\widehat x_k - \widehat x_{k-1})$ converges to $A'(x)$ by differentiability of $A$ and the first-order expansion of the exponential function. Further, dividing the difference equation by $\widehat x_k - \widehat x_{k-1}$ suggests that, in the continuous limit, the relaxed continuation revenue (call it $c$) is governed by:
\[
    c'(x) = A'(x)\bigl(1 - x c(x)\bigr),
    \qquad c(0)=1,
\]
which defines the function $c:[0,1)\to\mathbb R$ introduced in \cref{eq:reg-refined-ode}. Here, $c(x)$ represents the relaxed continuation revenue when the budget is allocated through infinitesimally small items rather than finitely many discrete items. The rate $A'(x)$ determines how much additional budget becomes available as the reciprocal cap increases, while the factor $1 - x c(x)$ accounts for the revenue generated by these additional items net of the continuation revenue they displace.

The limiting calculation motivates our choice of $c$, but does not by itself establish an upper bound for finite instances. We first establish the required properties of $A$ and $c$, showing that $c(x) \ge 1$ and that $c$ attains its maximum $1/x^\star$ at the unique point $x^\star$ satisfying $x^\star c(x^\star)=1$. We then prove that, starting from a continuation value of at least one, the continuous evolution over any budget-binding reciprocal-cap interval weakly dominates the corresponding single discrete step. Applying this comparison successively gives $\widehat C_k \le c(\widehat x_k)$ for every $k \in [N]$. Combining these results will bound the original cap-pricing revenue by $1/x^\star$.

First, we will need to establish a few properties of $A$ and $c$. As $A(x) = -\log(1 - x^2)/x$ for $x \in (0, 1)$ and with $A(0) = 0$, its continuous extension has derivative $A'(x) = \frac{2}{1-x^2} + \frac{\log(1-x^2)}{x^2}$ with $A'(0) = 1$, from which we can see that $A(x)$ is strictly increasing and invertible, and $\lim_{x \to 1} A(x) = \infty$. Further, we can upper bound $A'(x)$ for $x \in (0, 1)$ by $A'(x) \le \frac{2}{1-x^2} - 1$, where the inequality follows from the fact that $\log(1 - x^2) \le -x^2$ for $x \in [0, 1)$. Multiplying both sides by $(1 - x)^2$ gives
\[
    A'(x)(1 - x)^2 \le \frac{(1 - x)(1 + x^2)}{1 + x} \le 1 - x ,
\]
where the second inequality uses $x^2 \le x$. 

We next establish the required properties of $c$. As $c$ solves $c'(x) = A'(x)\bigl(1 - x c(x)\bigr)$ with boundary condition $c(0) = 1$, it is defined by a linear differential equation whose coefficients are continuous on every compact sub-interval of $[0,1)$, so it has a unique solution throughout $[0, 1)$. Since $A'(x) > 0$, the sign of $c'(x)$ is determined by $1 - x c(x)$. Define $g(x) := 1 - x c(x)$, we will next show that $g$ has exactly one zero in $(0, 1)$.

First, suppose that $g$ has no zero in $(0, 1)$. Since $g(0) = 1$ and $g$ is continuous, we would have $g(x) > 0$ for all $x \in [0, 1)$, or equivalently, $1/x > c(x)$ for all $x \in (0,1)$ by definition of $g$. However, since $A'(x) > 0$, this also implies that $c'(x) > 0$ and thus $c$ is strictly increasing. Since $c(0) = 1$ and it is strictly increasing, and since $1/x$ is strictly decreasing and equals 1 when $x = 1$, there must be an $x_0 \in [0, 1)$ such that $c(x_0) > 1/x_0$ contradicting the assumption. Thus, $g$ has at least one zero in $(0,1)$. We will now prove that this zero is unique.

To that end, note that $g'(x) = -c(x) - x c'(x)$. At any zero of $g$, we have $c(x) = 1/x$ by definition of $g$ and $c'(x) = 0$ by the defining ODE, and hence $g'(x)=-1/x<0$. Thus, every zero is a strict crossing from positive to negative. Furthermore, there cannot be a second zero since after the first crossing, $g$ is negative, and its first subsequent return to zero would have a non-negative left derivative, contradicting the strictly negative derivative at every zero. Therefore, there is a unique $x^\star \in (0, 1)$ such that $x^\star c(x^\star) = 1$, with $g(x) > 0$ for $x \in (0, x^\star)$ and $g(x) < 0$ for $x \in (x^\star, 1)$. It follows that $c$ is strictly increasing up to $x^\star$ and strictly decreasing afterward. Moreover, $c(x) \ge c(0) = 1$ for $x \in (0, x^\star]$, while $g(x) < 0$ implies $c(x) > 1/x > 1$ for $x \in (x^\star, 1)$. Consequently,
\begin{equation}\label{eq:reg-refined-envelope-maximum}
    c(x) \ge 1 \quad \text{ for } x \in [0, 1),
    \qquad
    \max_{0 \le x < 1} c(x) = c(x^\star) = \frac1{x^\star}.
\end{equation}
Numerically solving the ODE which defines $c(x)$ and locating the unique root of $x c (x) = 1$ gives $x^\star \approx 0.698636$, and hence the maximum value is $c(x^\star) = 1/x^\star \approx 1.43136$.

We now show that the continuous evolution dominates a single budget-binding step of the discrete recursion. Fix $y\in[0,1)$ and an incoming continuation value $z\ge1$, and let $\widehat c:[y, 1) \to \mathbb R$ solve:
\[
    \widehat c'(x) = A'(x)\bigl(1 - x\widehat c(x)\bigr), \qquad \widehat c(y) = z.
\]
As above, this solution exists and is unique throughout $[y, 1)$. For $x \in [y, 1)$, define $a(x) := \exp\left(A(x) - A(y)\right) - 1$ and $d(x) := \frac{a(x) + z}{1 + x a(x)}$. For $x > y$, $\log(1 + a(x)) = A(x) - A(y)$ is the increment in the cumulative ``budget,'' and $d(x)$ is the continuation revenue obtained by adding an item with distribution $\Tri(a(x), 1/x)$ priced at its cap. We will prove that $d(x) \le \widehat c(x)$, for $x \in [y, 1)$.

To compare $d(x)$ and $\widehat c(x)$, we show that an appropriately weighted difference between them is non-decreasing. Specifically, define the positive integrating factor $\mu(x) := \exp\left(\int_y^x t A'(t) dt\right)$ and the weighted difference $h(x) := \mu(x)\bigl(\widehat c(x) - d(x)\bigr)$. Since $d(y) = z = \widehat c(y)$, we have $h(y) = 0$. Further, since $\mu(x) > 0$, it suffices to show that $h'(x) \ge 0$ for $x \in [y, 1)$. The definitions of $a$ and $\mu$ give:
\[
    a'(x)=A'(x)\bigl(1+a(x)\bigr), \qquad \mu'(x) = x A'(x) \mu(x).
\]
Using the differential equation for $\widehat c$, we obtain:
\begin{align}\label{eq:reg-refined-weighted-gap-intermediate}
    h'(x)
    &= \mu(x)\left(\widehat c'(x) - d'(x) + x A'(x)\bigl(\widehat c(x) - d(x)\bigr) \right) \notag\\
    &= \mu(x)\left(A'(x)\bigl(1 - x d(x)\bigr) -d'(x)\right).
\end{align}
Differentiating $d(x)$ gives,
\[
    d'(x) = \frac{A'(x)(1 + a(x))(1 - x z) - a(x)(a(x) + z)}{(1 + x a(x))^2}.
\]
Substituting this expression into \cref{eq:reg-refined-weighted-gap-intermediate} and simplifying yields:
\begin{equation}\label{eq:reg-refined-weighted-gap-derivative}
    h'(x)
    =
    \mu(x)\frac{a(x)}{(1+xa(x))^2}
    \left(
        a(x)+z-A'(x)(1-x)(1-xz)
    \right).
\end{equation}
Since $z \ge1$, we have $1 - x z \le 1 - x$. Together with $A'(x) \ge 0$ and our earlier upper bound on $A'(x)$, this implies that,
\[
    A'(x)(1-x)(1 - x z) \le A'(x)(1 - x)^2 \le 1 \le a(x) + z,
\]
where the last inequality follows from $a(x) \ge 0$ and $z \ge 1$. Thus, the expression in parentheses in \cref{eq:reg-refined-weighted-gap-derivative} is non-negative, as are all remaining factors, implying $h'(x) \ge 0$ and $h(x) \ge h(y) = 0$. Since $\mu(x) > 0$, this proves $d(x) \le \widehat c(x)$, for $x \in [y, 1)$.

We can now show $\widehat C_k \le c(\widehat x_k)$ for all $k$ by induction on $k$. The base case is $\widehat C_0 = 1 = c(0) = c(\widehat x_0)$. Now suppose, for some $k \in [N]$ that $\widehat C_{k-1} \le c(\widehat x_{k-1})$. Since the recursion is increasing in its incoming continuation value,
\[
    \widehat C_k = \frac{a_k+\widehat C_{k-1}}{1+a_k\widehat x_k} \le \frac{a_k + c(\widehat x_{k-1})}{1 + a_k\widehat x_k}.
\]
By \cref{eq:reg-refined-envelope-maximum}, $c(\widehat x_{k-1}) \ge 1$, and we may apply the result that $d(x) \le \widehat c(x)$ with $y = \widehat x_{k-1}$, $x = \widehat x_k$, and $z = c(\widehat x_{k-1})$. The identity $a_k = \exp\left(A(\widehat x_k) - A(\widehat x_{k-1})\right) - 1$ ensures that $d(\widehat x_k)$ is exactly the discrete expression on the right-hand side above. Moreover, by uniqueness of the ODE solution, the continuous evolution starting from $c(\widehat x_{k-1})$ coincides with $c$ on $[\widehat x_{k-1},1)$. Hence,
\[
    \frac{a_k+c(\widehat x_{k-1})}{1+a_k\widehat x_k} \le c(\widehat x_k),
\]
which completes the induction.

Finally, combining $C_N\le\widehat C_N$, the induction above, and \cref{eq:reg-refined-envelope-maximum}, we obtain
\[
    \rev_{DP}(P,\v_1, \ldots, \v_N) = C_N \le \widehat C_N \le c(\widehat x_N) \le \max_{0 \le x< 1} c(x) = \frac{1}{x^\star},
\]
as desired.
\hfill\Halmos\end{proof}

%%%%%%%%%%%%%%%%%%%%%%%%%%%%%%%%%%%%%%%%%%%%%
\vskip 10pt
%%%%%%%%%%%%%%%%%%%%%%%%%%%%%%%%%%%%%%%%%%%%%

\begin{proof}{Proof of \cref{lem:lower_bound}.}
For any fixed $\tilde x \in (0, 1)$, we will construct a sequence of canonical instances with optimal uniform-pricing revenue exactly one and cap-pricing revenue converging to $c(\tilde x)$. The lemma then follows by taking $\tilde x = x^\star$, since the proof of \cref{lem:reg-refined-Tri-bound} establishes that $c(x^\star) = 1/x^\star$.

Specifically, for each $m \ge 1$, we will build a canonical instance consisting of a terminal item with independent valuation $V_0 \sim \Tri(1,P)$ and $m$ additional items valued independently as $V_k \sim \Tri(a_k, \v_k)$, $k \in [m]$, where $P$ is sufficiently large. Define $h_m := A(\tilde x)/m$ and assign all additional items the same parameter $a_k := h_m$. We will choose their caps so that the original uniform-pricing constraint binds at each cap price. Unlike the relaxed constraints used in the proof of \cref{lem:reg-refined-Tri-bound}, these constraints directly control the uniform-pricing revenue of the constructed instance.

To determine these caps, for each $k \in [m]$, define $G_k(x) := (1 - x^2)(1 + h_m x)^k$ for $x \in [0,1)$, and let $x_k \in (0, 1)$ be the unique positive solution to $G_k(x) = 1$. To see that such solutions exist, are unique, and are ordered, first note that differentiating $G_k$ gives:
\[
    G_k'(x) = (1 + h_m x)^{k-1}\left[k h_m - 2x - (k + 2)h_m x^2\right].
\]
The first factor is strictly positive, while the expression in brackets is strictly decreasing on $[0,1]$, positive at zero, and negative at one. Thus, $G_k'$ changes sign exactly once, from positive to negative, so $G_k$ first strictly increases and then strictly decreases. Since $G_k(0) = 1$ and $\lim_{x \to 1} G_k(x) = 0$, the equation $G_k(x) = 1$ has a unique positive solution $x_k \in (0,1)$, and $G_k(x) \le 1$ for every
$x \in [x_k, 1)$.

Furthermore, for all $x \in (0, 1)$ and $k \in [m-1]$, $G_{k+1}(x) = G_k(x)(1 + h_m x) > G_k(x)$. In particular, $G_{k+1}(x_k) > G_k(x_k) = 1$, so its positive solution $x_{k+1}$ lies to the right of $x_k$. Consequently, $0 < x_1 < \ldots < x_m < 1$. Set $\v_k := 1/x_k$ for each additional item $k \in [m]$ and choose $P > \v_1$. The caps then satisfy $P > \v_1 > \ldots > \v_m > 1$, as required for a canonical instance.

Next, we show that the optimal uniform-pricing revenue of the constructed instance is exactly one. First, consider any price $p \in (1, \v_1]$, let $x := 1/p$, and suppose that exactly the first $k$ additional items have caps at least $p$. These items and the terminal item are the only items that can be purchased at this price. The terminal item is rejected with probability $1/(1 + 1/p)$, and each of the first $k$ additional items is rejected with probability $1/(1 + h_m x)$. Thus, by independence,
\[
    \rev_{UP}(p) = \frac{1}{x}\left(1 - \frac{1}{(1 + x)(1 + h_m x)^k}\right).
\]
Rearranging the above gives $\rev_{UP}(p)\le1$ if and only if $G_k(1/p) = (1-p^{-2})(1 + h_m/p)^k \le 1$. Since $\v_k \ge p > 1$, we have $1/p \in [x_k, 1)$. The previously established bound $G_k(x) \le 1$ for $x \in [x_k, 1)$ therefore implies $\rev_{UP}(p) \le 1$ for every $p \in (1, \v_1]$. Moreover, at each cap price $p = \v_k$, the defining equality $G_k(x_k) = 1$ gives $\rev_{UP}(\v_k) = 1$. For $p \in (\v_1, P]$, only the terminal item can be purchased, yielding uniform-pricing revenue $p/(p+1) < 1$. For $p \in [0,1]$, the bound follows immediately from $\rev_{UP}(p) \le p \le 1$, while prices above $P$ earn zero because they exceed every item's cap. Hence, the constructed instance satisfies $\rev_{UP} = 1$.

It remains to show that the cap-pricing revenue approaches $c(\tilde x)$ as the number of items and the terminal cap grow. Let $C_k^{(P, m)}$ denote the continuation revenue after incorporating the terminal item and the first $k$ of $m$ additional items in decreasing-cap order. Then,
\[
    C_0^{(P, m)} = \frac{P}{P + 1},
    \qquad
    C_k^{(P, m)} = \frac{h_m + C_{k-1}^{(P, m)}}{1 + h_m x_k}, \qquad k\in[m].
\]
Let $C_k^{(\infty, m)}$ denote the same recursion with initial continuation revenue one, i.e.,
\[
    C_0^{(\infty, m)} = 1,
    \qquad
    C_k^{(\infty, m)} = \frac{h_m + C_{k-1}^{(\infty, m)}}{1 + h_m x_k}, \qquad k \in [m].
\]
Subtracting the two recursions gives $C_k^{(\infty, m)} - C_k^{(P, m)} = \frac{C_{k-1}^{(\infty, m)} - C_{k-1}^{(P, m)}}{1 + h_m x_k}$ for any $k \in [m]$. The initial difference is $C_0^{(\infty, m)} - C_0^{(P, m)} = 1/(P + 1)$. Since $1 + h_m x_k > 1$, each subsequent difference remains non-negative and cannot exceed the preceding difference. Thus, iterating gives $0 \le C_k^{(\infty, m)} - C_k^{(P, m)} \le \frac{1}{P+1}$ for all $k$. Thus it suffices to show that $C_m^{(\infty, m)}\to c(\tilde x)$ as $m \to \infty$.

To prove this convergence, we first show that the reciprocal caps $x_k$ are close to the grid $A^{-1}(k h_m)$ in the transformed coordinate $s = A(x)$. We then solve the continuation recursion explicitly and show that its products and sums converge to the corresponding continuous integrals. All $O(\cdot)$ terms below are taken as $h_m \to 0$, with constants that may depend on $\tilde x$ but are independent of $m$ and $k$.

To begin, taking logarithms of $G_k(x_k)=1$ and dividing by $x_k$ gives $A(x_k) = -\frac{\log(1 - x_k^2)}{x_k} = k\frac{\log(1 + h_m x_k)}{x_k}$. Since $\log(1 + u) \le u$ for $u \ge 0$, we have $A(x_k) \le k h_m \le A(\tilde x)$. As shown in the proof of \cref{lem:reg-refined-Tri-bound}, $A$ is strictly increasing, so $x_k \le \tilde x$. Moreover, applying $\log(1 + u)\ge u - u^2/2$ gives
\[
    0 \le k h_m - A(x_k)
    = \frac{k}{x_k} \left(h_m x_k - \log(1 + h_m x_k)\right)
    \le \frac{k h_m^2 x_k}{2}
    \le \frac{A(\tilde x)\tilde x}{2} h_m.
\]
Hence, $\max_{k \in [m]}\left|A(x_k) - k h_m\right| = O(h_m)$. Since $A'$ is continuous and strictly positive on $[0, \tilde x]$, $A^{-1}$ is Lipschitz on $[0, A(\tilde x)]$. Applying this Lipschitz bound to $A(x_k)$ and $k h_m$ yields:
\begin{equation}\label{eq:reg-refined-tight-grid}
    x_k = A^{-1}(k h_m) + O(h_m),
\end{equation}
uniformly in $k$. Thus, although the reciprocal caps are not equally spaced, their images under $A$ lie within $O(h_m)$ of the equally spaced grid $k h_m$, which explains the use of the transformed coordinate $s = A(x)$.

We next turn to the continuation recursion. We remove the successive denominators by multiplying the continuation recursion by their cumulative product, allowing the resulting identities to telescope. Define $D_0 := 1$ and $D_k := \prod_{j = 1}^k(1 + h_m x_j)$ for $k \in [m]$. Note that the reciprocal caps $x_k$ and the products $D_k$ also depend on $m$; however, we suppress this dependence in the notation. Now, multiplying the recursion for $C_k^{(\infty, m)}$ by $D_k = D_{k-1}(1 + h_m x_k)$ gives $D_k C_k^{(\infty, m)} = D_{k-1} C_{k-1}^{(\infty, m)} + h_m D_{k-1}$. Summing this identity from $k=1$ to $m$, and using $C_0^{(\infty, m)} = D_0 = 1$, yields
\begin{equation}\label{eq:reformulation_C_m}
    C_m^{(\infty, m)} = \frac{1 + h_m \sum_{k=0}^{m-1} D_k}{D_m}.
\end{equation}
It remains to identify the limits of the numerator and denominator. Let $\Phi(s) := \int_0^s A^{-1}(u) du$ for $s \in [0, A(\tilde x)]$. Using \cref{eq:reg-refined-tight-grid}, the bound $x_j \le \tilde x$, and that $\log(1 + u) = u + O(u^2)$, we have
\begin{align}\label{eq:sum_D_K}
    \log D_k
    = \sum_{j=1}^k\log(1 + h_m x_j)
    = h_m \sum_{j=1}^k x_j + O(k h_m^2)
    = h_m \sum_{j=1}^k A^{-1}(j h_m) + O(h_m),
    \qquad k\in[m],
\end{align}
uniformly in $k$, where the last equality uses $k h_m \le A(\tilde x)$. The final sum is a Riemann sum for $\Phi(kh_m)$. Because $A^{-1}$ is Lipschitz on $[0, A(\tilde x)]$, its approximation error is $O(h_m)$ uniformly in $k$. Thus, $\log D_k = \Phi(k h_m) + O(h_m)$. Since $\Phi$ is bounded on $[0, A(\tilde x)]$, exponentiating gives $D_k = e^{\Phi(k h_m)} + O(h_m)$ uniformly for $k = 0, \ldots, m$. In particular, $D_m \to e^{\Phi(A(\tilde x))}$. Because the error is uniform in $k$, its contribution to $h_m\sum_{k=0}^{m-1}D_k$ is $O(mh_m^2)=O(h_m)$, which vanishes as $m\to\infty$. Applying the Riemann-sum argument again gives
\[
    h_m \sum_{k=0}^{m - 1} D_k \to \int_0^{A(\tilde x)} e^{\Phi(s)}\,ds.
\]
Substituting these limits into \cref{eq:reformulation_C_m}, we obtain:
\begin{equation}\label{eq:convergence_C_m}
    C_m^{(\infty, m)} \to e^{-\Phi(A(\tilde x))}\left(1 + \int_0^{A(\tilde x)} e^{\Phi(s)} ds\right).
\end{equation}

Finally, we identify this limit with $c(\tilde x)$. Define $y(s):=c(A^{-1}(s))$ for $s\in[0,A(\tilde x)]$. By the chain rule and the defining ODE $c'(x)=A'(x)(1-xc(x))$, the change of variables $s=A(x)$ gives
\[
    y'(s) = 1 - A^{-1}(s)y(s), \qquad y(0) = 1.
\]
Because $\Phi'(s) = A^{-1}(s)$, the integrating factor is $e^{\Phi(s)}$, and hence $\frac{d}{ds} \left(e^{\Phi(s)} y(s)\right) = e^{\Phi(s)}$. Integrating this identity from $0$ to $A(\tilde x)$ and using $\Phi(0) = 0$ and $y(0) = 1$ yields
\[
    y(A(\tilde x)) =e^{-\Phi(A(\tilde x))} \left(1 + \int_0^{A(\tilde x)} e^{\Phi(s)} ds\right).
\]
Since $A^{-1}(A(\tilde x)) = \tilde x$, we have $y(A(\tilde x)) = c(A^{-1}(A(\tilde x))) = c(\tilde x)$. Comparing this expression with \cref{eq:convergence_C_m} proves that $C_m^{(\infty, m)}\to c(\tilde x)$.

Taking $\tilde x = x^\star$, we have $c(x^\star) = 1/x^\star$. For any $\ep > 0$, first choose $m$ sufficiently large that $C_m^{(\infty, m)} > 1/x^\star - \ep/2$, and then choose $P > \v_1$ sufficiently large that $1/(P+1) < \ep/2$. The bound on the terminal approximation gives
\[
    \rev_{DP} \ge C_m^{(P, m)} \ge C_m^{(\infty, m)} - \frac{1}{P + 1} > \frac{1}{x^\star} - \ep.
\]
The first inequality holds because cap-pricing is a feasible differentiated price vector. Combining this with $\rev_{UP} = 1$ completes the proof.
\hfill\Halmos\end{proof}

\subsection{Additional Omitted Proofs}
\label{sec:omitted_pfs_2}
In this subsection, we state and prove the auxiliary lemmas used in the preceding proofs.

\begin{lemma}[Rejection Probability Bound]
\label{lem:reg-refined-prefix-rejection}
Consider $m$ items with independent triangular-revenue valuations $V_i \sim \Tri(a_i,\v_i)$, $i \in [m]$, whose caps satisfy $\v_1 \ge \ldots \ge \v_m \ge t$ for some $t > 0$. Suppose customers are min-buying and each item is priced at its cap. Let $C_i$ denote the continuation revenue after incorporating the first $i$ items in decreasing-cap order, so that $C_0 := 0$ and $C_i := \frac{a_i + C_{i-1}}{1 + a_i/\v_i}$ for all $i \in [m]$. Let $Q_m(t) := \prod_{i=1}^m \frac{t}{t + a_i}$ be the probability that all items are rejected at the uniform price $t$. Then,
\begin{equation}\label{eq:reg-refined-prefix-rejection-bound}
    Q_m(t) \le \frac{t}{t + C_m}.
\end{equation}
\end{lemma}
\begin{proof}{Proof of \cref{lem:reg-refined-prefix-rejection}.}
We prove \cref{eq:reg-refined-prefix-rejection-bound} by induction on the number of items. The claim is immediate for $m = 0$, since $C_0 = 0$ and $Q_0(t) = 1$, where the empty product is defined to be one. Now suppose the claim holds for the first $m - 1$ items, i.e., $Q_{m-1}(t) \le t/(t + C_{m-1})$. Adding item $m$ gives continuation revenue, $C_m = \frac{a_m + C_{m-1}}{1 + a_m/\v_m}$ and total rejection probability $Q_m(t) = Q_{m-1}(t)t/(t+a_m)$. By the induction hypothesis,
\begin{equation}\label{eq:induction_m}
    Q_m(t) \le \frac{t}{t + C_{m-1}}\frac{t}{t + a_m}.
\end{equation}
It remains to show that the right-hand side of \cref{eq:induction_m} is at most $t/(t+C_m)$. Because $t > 0$ and all denominators are positive, this is equivalent to showing $t(t + C_m) \le (t + C_{m-1})(t + a_m)$. By the definition of continuation revenue,
\[
    C_m = \frac{a_m + C_{m-1}}{1 + a_m/\v_m} \le a_m + C_{m-1}.
\]
Hence, the difference between the right- and left-hand sides satisfies:
\[
    (t + C_{m-1})(t + a_m) - t(t + C_m) = t(a_m + C_{m-1} - C_m) + a_m C_{m-1} \ge 0,
\]
where we also use $C_{m-1}\ge0$. Combining this inequality with \cref{eq:induction_m} gives $Q_m(t) \le t/(t + C_m)$, completing the induction.
\hfill \Halmos \end{proof}

%%%%%%%%%%%%%%%%%%%%%%%%%%%%%%%%%%%%%%%%%%%%%
\vskip 10pt
%%%%%%%%%%%%%%%%%%%%%%%%%%%%%%%%%%%%%%%%%%%%%

\begin{lemma}[Continuous Approximation of Triangular Instances] \label{lem:continuous_approx}
Consider a finite instance $\cI$ consisting of $m$ items with independent triangular-revenue valuations $V_i \sim \Tri(a_i, \v_i)$, $i \in [m]$, and suppose customers are min-buying. Let $\p =(\v_1, \ldots, \v_m)$ be its cap-pricing vector. If $\rev_{UP}(\cI) \le 1$, then, for every sufficiently small $\delta \in (0, 1)$, there exists an independent continuous regular instance $\cI^\delta$ such that
\[
    \rev_{DP}(\p;\cI^\delta) = \rev_{DP}(\p;\cI),
    \qquad
    \rev_{UP}(\cI^\delta) \le \frac{1}{1-\delta}.
\]
\end{lemma}

\begin{proof}{Proof of \cref{lem:continuous_approx}.}
Consider an item with $V_i\sim\Tri(a_i,\v_i)$ in the triangular instance $\cI$, and let $q_i:=a_i/(a_i+\v_i)$ be its acceptance probability at its cap $\v_i$. By the definition of a triangular-revenue distribution, its revenue-quantile function is
\[
    \rev_{\Tri(a_i, \v_i)}(q)
    =
    \begin{cases}
        \v_i q, & 0 \le q \le q_i,\\
        a_i(1 - q), & q_i \le q \le 1.
    \end{cases}
\]
Choose $0<\delta<1$ sufficiently small that $\delta \le a_i+\v_i$ for every $i \in [m]$, and define the perturbed revenue-quantile function
\[
    \rev_i^\delta(q)
    =
    \begin{cases}
        \v_i q + \delta q \left(1 - \dfrac{q}{q_i}\right), & 0 \le q \le q_i,\\
        a_i(1 - q), & q_i \le q \le1.
    \end{cases}
\]
The two branches agree at $q_i$. The first branch is strictly concave, and the second branch is linear. Moreover, the left and right derivatives at $q_i$ satisfy
\[
    (\rev_i^\delta)'(q_i^-) = \v_i-\delta, \qquad (\rev_i^\delta)'(q_i^+) = -a_i.
\]
Because $\delta\le a_i+\v_i$, the slope weakly decreases at the joining point $q_i$. Thus, $\rev_i^\delta$ is concave.

The upper quantile function induced by $\rev_i^\delta$ is
\[
    Q_i^\delta(q)
    = \frac{\rev_i^\delta(q)}{q}
    =
    \begin{cases}
        \v_i + \delta\left(1-\dfrac{q}{q_i}\right), & 0 < q \le q_i,\\
        \dfrac{a_i(1 - q)}{q}, & q_i \le q \le 1.
    \end{cases}
\]
Since $Q_i^\delta$ is continuous and strictly decreasing on $(0,1]$, we can invert it to obtain the CDF
\[
    F_i^\delta(x)
    =
    \begin{cases}
        0, & x<0,\\
        \dfrac{x}{a_i+x}, & 0\le x \le \v_i,\\
        1 - \dfrac{q_i}{\delta}(\v_i + \delta - x), & \v_i \le x \le \v_i + \delta,\\
        1, & x \ge \v_i + \delta.
    \end{cases}
\]
The branches agree at their endpoints, so $F_i^\delta$ is continuous. Since its revenue-quantile function $\rev_i^\delta$ is concave, the induced distribution is regular. Replacing each item in $\cI$ by an independent valuation $V_i^\delta \sim F_i^\delta$ therefore gives an independent continuous regular instance $\cI^\delta$.

Next, we show that the replacement preserves the revenue of the original cap-pricing vector. Since $Q_i^\delta(q_i) = \v_i$ and $F_i^\delta$ is continuous, we have
\[
    \Pr(V_i^\delta \ge \v_i) = q_i = \Pr(V_i \ge \v_i), \qquad i\in[m].
\]
Thus, pricing every item at its original cap preserves each acceptance probability. Independence and the ordering of the designated prices are also unchanged, so min-buying revenue is preserved exactly:
\[
    \rev_{DP}(\p;\cI^\delta) = \rev_{DP}(\p;\cI).
\]

It remains to bound uniform-pricing revenue. The construction leaves the probability mass below $\v_i$ unchanged and spreads the atom at $\v_i$ over $[\v_i,\v_i+\delta]$. Thus, the original and
perturbed valuations can be coupled so that $V_i \le V_i^\delta \le V_i + \delta$. Consequently, for every $p > \delta$,
\[
    \Pr(V_i^\delta < p) \ge \Pr(V_i < p - \delta).
\]
We use strict rejection probabilities here because the original triangular-revenue distributions have atoms at their caps. By independence, for every $p > 1$,
\begin{align*}
    \rev_{UP}(p;\cI^\delta)
    &= p\left(1-\prod_{i=1}^m\Pr(V_i^\delta<p)\right)\\
    &\le p\left(1-\prod_{i=1}^m\Pr(V_i<p-\delta)\right)\\
    &= \frac{p}{p-\delta}\rev_{UP}(p-\delta;\cI)\\
    &\le \frac{p}{p-\delta}\\
    &\le \frac{1}{1-\delta},
\end{align*}
where the last two inequalities follow from $\rev_{UP}(\cI)\le1$, $p>1$, and $0<\delta<1$. For $0\le p\le1$, uniform-pricing revenue is at most $p\le1$. Taking the supremum over all nonnegative uniform prices therefore gives $\rev_{UP}(\cI^\delta)\le1/(1-\delta)$, completing the proof.
\hfill\Halmos\end{proof}

\end{APPENDICES}

% \newpage
% \bibliographystyle{ormsv080}
% \bibliography{bibliography.bib}

%%%%%%%%%%%%%%%%%
\end{document}